\documentclass[10pt]{article}  
\usepackage[english]{babel}
\usepackage[toc,page]{appendix}
\usepackage{bbm}
\usepackage{fullpage}
\usepackage[top=1in, bottom=1.25in, left=1in, right=1in]{geometry}
\usepackage{etaremune}
\usepackage{enumitem}

\usepackage{witharrows}

\usepackage{amssymb,amsmath,amsthm}

\usepackage{hyperref}

\DeclareSymbolFont{rsfs}{U}{rsfs}{m}{n}
\DeclareSymbolFontAlphabet{\mathscrsfs}{rsfs}
\usepackage{mathrsfs}

\usepackage{url}

\hypersetup{
  colorlinks   = true, 
  urlcolor     = blue, 
  linkcolor    = blue, 
  citecolor   = red 
}

\usepackage[labelformat=empty]{caption}

\newtheorem{theorem}{Theorem}[section]
\newtheorem{lemma}[theorem]{Lemma}

\newtheorem{proposition}[theorem]{Proposition}
\newtheorem{corollary}[theorem]{Corollary}

\theoremstyle{definition}
\newtheorem{definition}{Definition}
\newtheorem{remark}[theorem]{Remark}

\numberwithin{equation}{section}

\usepackage{times}

\newcommand{\bea}{\begin{eqnarray}}
\newcommand{\eea}{\end{eqnarray}}
\newcommand{\<}{\langle}
\renewcommand{\>}{\rangle}

\newcommand{\wt}{\widetilde}
\newcommand{\op}{\text{op}}
\newcommand{\wh}{\widehat}

\def\Ito{It\^{o}}

\def\ie{\text{i.e.~}}
\def\ae{\text{a.e.~}}
\newcommand\eg{{\text{\eg~}}}
\def\iid{\text{i.i.d.~}}

\def\eps{{\varepsilon}}

\def\ind{{\mathbbm 1}}

\def\bW{{\boldsymbol{W}}}

\def\bM{{\boldsymbol{M}}}

\def\bx{{\boldsymbol{x}}}

\def\cF{{\mathcal F}}
\def\cG{{\mathcal G}}

\def\cE{{\mathcal E}}

\def\op{{\rm op}}

\def\vzero{\vec 0}

\def\bsig{{\boldsymbol {\sigma}}}

\def\brho{{\boldsymbol \rho}}

\def\by{{\boldsymbol y}}

\def\mod{{\mathsf{mod}}}

\def\bv{{\boldsymbol{v}}}
\def\bz{{\boldsymbol{z}}}
\def\bx{{\boldsymbol{x}}}

\def\bb{{\boldsymbol{b}}}

\def\bB{\boldsymbol{B}}

\def\bT{\boldsymbol{T}}

\def\de{{\rm d}}

\def\bW{\boldsymbol{W}}

\def\<{\langle}
\def\>{\rangle}
\def\Tr{{\sf Tr}}

\def\cN{{\cal N}}

\def\by{{\boldsymbol{y}}}

\def\cE{{\mathcal E}}

\def\b0{{\boldsymbol{0}}}

\def\rd{{\mathrm {rad}}}

\def\bG{{\boldsymbol G}}

\DeclareMathOperator*{\plim}{p-lim}

\def\cA{{\mathcal A}}
\def\cI{{\mathcal I}}
\def\cS{{\mathcal S}}

\def\cB{{\mathcal B}}

\renewcommand{\b}{\mathbf{b}}

\def\lt{\left}
\def\rt{\right}

\def\la{\langle}
\def\ra{\rangle}

\def\eps{\varepsilon}

\def\bbE{{\mathbb{E}}}

\def\bbP{{\mathbb{P}}}
\def\bbR{{\mathbb{R}}}

\def\bbZ{{\mathbb{Z}}}

\def\cA{{\mathcal{A}}}
\def\cB{{\mathcal{B}}}
\def\cF{{\mathcal{F}}}

\def\cN{{\mathcal{N}}}
\def\cP{{\mathcal{P}}}

\def\sH{{\mathscr{H}}}

\def\bM{\mathbf{M}}

\def\ALG{{\mathsf{ALG}}}

\def\sph{\mathrm{sp}}

\newcommand{\norm}[1]{{\lt\|#1\rt\|}}

\newcommand{\oC}{\overline{C}}

\usepackage[suppress]{color-edits}  
\addauthor{rev}{blue}   
\addauthor{ms}{red}

\author{Mark Sellke}

\title{The Threshold Energy of Low Temperature Langevin Dynamics \\ for Pure Spherical Spin Glasses}

\date{}

\begin{document}

\maketitle

\begin{abstract}
    We study the Langevin dynamics for spherical $p$-spin models, focusing on the short time regime described by the Cugliandolo--Kurchan equations.
    Confirming a prediction of \cite{cugliandolo1993analytical}, we show the asymptotic energy achieved is exactly $E_{\infty}(p)=2\sqrt{\frac{p-1}{p}}$ in the low temperature limit.
    The upper bound uses hardness results for Lipschitz optimization algorithms and applies for all temperatures.
    For the lower bound, we prove the dynamics reaches and stays above the lowest energy of any \emph{approximate local maximum}.
    In fact the latter behavior holds for any Hamiltonian obeying natural smoothness estimates, even with disorder-dependent initialization and on exponential time-scales. 
\end{abstract}

\section{Introduction}

Fix an integer $p\geq 2$ and let $\bG_N^{(p)} \in \lt(\bbR^N\rt)^{\otimes p}$ be an order $p$ tensor with \iid standard Gaussian entries $g_{i_1,\dots,i_p}\sim \cN(0,1)$. The pure $p$-spin Hamiltonian $H_N:\bbR^N\to\bbR$ is the random homogeneous polynomial
\begin{equation}
    \label{eq:def-hamiltonian}
    H_N(\brho) = N^{-\frac{p-1}{2}} \la \bG_N^{(p)}, \brho^{\otimes p} \ra
    =
    N^{-\frac{p-1}{2}}
    \sum_{i_1,\dots,i_p=1}^N g_{i_1,\dots,i_p}\rho_{i_1}\dots\rho_{i_p}.
\end{equation}
We study the Langevin dynamics for this Hamiltonian on the sphere 
\[
\cS_N\equiv\Big\{\brho \in \bbR^N : \sum_{i=1}^N \brho_i^2 = N\Big\}
\] 
which is described, \revedit{for $\bx_0\in\cS_N$}, by the stochastic differential equation
\begin{equation}
    \label{eq:langevin-dynamics}
    \de \bx_t 
    =
    \lt(\beta\nabla_{\sph} H_N(\bx_t) 
    - 
    \frac{(N-1)\bx_t}{2N}\rt)
    \de t 
    + 
    P_{\bx_t}^{\perp}~\de \bB_t.
\end{equation}
Here $P_{\bx}^{\perp}$ is the projection matrix onto the orthogonal complement of $\bx$, $\nabla_{\sph}H_N(\bx)=P_{\bx}^{\perp} \nabla H_N(\bx)$ is the Riemannian gradient on $\cS_N$, and $\bB_t$ is standard Brownian motion on $\bbR^N$.
\revedit{It is easy to verify that $\bx_t\in\cS_N$ almost surely for all $t$ by applying \Ito's formula to $\|\bx_t\|^2$.}

The unique invariant measure for \eqref{eq:langevin-dynamics} on $\cS_N$ is the random Gibbs measure ${\de \mu_{2\beta,H_N}(\bx)\propto e^{2\beta H_N(\bx)}~\,\de \bx}$, where $\de\bx$ denotes the uniform measure on $\cS_N$.\footnote{ It is common to include a factor $\sqrt{2}$ with the diffusive term so the values of $\beta$ match. We follow the convention of \cite{ben2006cugliandolo,dembo2007limiting} so the Cugliandolo--Kurchan equations take the same form.}
The static behavior of this and other spin glass Gibbs measures have been vigorously studied since the seminal works \cite{sherrington1975solvable,parisi1979infinite}. Notably, this led to the proof of the Parisi formula for the free energy in \cite{talagrand2006parisi,talagrand2006spherical} and in more generality by \cite{panchenko2013parisi,panchenko2014parisi,chen2013aizenman,auffinger2017parisi}.
For the specific models \eqref{eq:def-hamiltonian}, more precise results are available at sufficiently low temperature due to \cite{subag2017extremal,subag2017geometry}.

The \emph{dynamical} behavior of mean-field spin glasses and related models has also received significant attention from both physicists and mathematicians --- see the surveys \cite{bouchaud1998out,cugliandolo2004course,cugliandolo2023recent} and \cite{arous2002aging,guionnet2007dynamics,jagannath2019dynamics} respectively.
While fast mixing is known to hold at high-temperature \cite{gheissari2019spectral,eldan2022spectral,anari2022entropic,adhikari2022spectral}, low temperature dynamics are expected to exhibit \emph{aging} in which the system's effective memory length grows with time.
Much work has studied and rigorously established this behavior, e.g. on ``activated'' timescales which are exponentially large yet shorter than the mixing time \cite{bouchaud1995aging,ben2006dynamics,ben2008arcsine,ben2008universality,ben2012universality}.

Our primary focus will be the dynamical behavior on shorter time scales which are dimension-free, as first studied in \cite{sompolinsky1982relaxational}.
Already many interesting behaviors are expected, including a short-time version of aging as well as convergence to a \emph{threshold energy} different from the ground state, since the important works \cite{crisanti1993spherical,cugliandolo1993analytical,biroli1999dynamical}.
On the sphere, a primary tool was the \emph{Cugliandolo--Kurchan equations}: a closed system of integro-differential equations for certain two-time observables. 
These equations were proved for a soft spherical analog of \eqref{eq:langevin-dynamics} in \cite{ben2006cugliandolo}, 
and many other rigorous works on this topic have appeared including \cite{grunwald1996sanov,arous1997symmetric,guionnet1997averaged,guionnet2005long,zamfir2008limiting,dembo2020dynamics,dembo2021diffusions,maclaurin2021emergent,celentano2021high}.
Unfortunately, rigorously analyzing the solutions to these exact equations beyond existence and uniqueness has proven elusive except in the case $p=2$ or at high temperature \cite{arous2001aging,dembo2007limiting}.
Indeed \cite{ben2020bounding} lists several basic questions which seem difficult to answer from this point of view.
\paragraph{The Threshold Energy $E_{\infty}$}

Let us now sharpen our focus to the question we address: the energy achieved by Langevin dynamics.
In other words, we will investigate its performance for \emph{optimization}.
It was predicted since \cite{cugliandolo1993analytical} \revedit{(see the formula for $\cE_{0c}$ therein)} that when $T$ and $\beta$ are large constants, $H_N(\bx_T)/N\approx E_{\infty}(p)\equiv 2\sqrt{\frac{p-1}{p}}$ holds.
$E_{\infty}$ is strictly smaller than the Parisi ground state energy for all $p\geq 3$. However it is a natural threshold for local landscape reasons: typical critical points of $H_N$ with energy below $E_{\infty}$ are saddle points, while those above $E_{\infty}$ are typically local maxima (see e.g.\ \cite{auffinger2013complexity,auffinger2013random,subag2017complexity,subag2021concentration} for much more detailed results of this kind).
While this is highly suggestive, it was unclear how to deduce rigorous consequences for the Langevin dynamics.

Further evidence for this prediction has come from the recent works \cite{subag2018following,huang2021tight,huang2023algorithmic} which showed that $E_{\infty}$ is the largest attainable energy for a natural class of \emph{Lipschitz algorithms}, \ie functions $\cA:H_N\mapsto \bx$ with dimension-free Lipschitz constant and for which $\|\cA(H_N)\|\leq \sqrt{N}$ holds almost surely. The relevant result is recalled in Proposition~\ref{prop:BOGP} below (see also \cite{montanari2021optimization,ams20,sellke2021optimizing} for related algorithms in the Ising case).
Indeed a close variant of \eqref{eq:langevin-dynamics} with reflecting boundary conditions inside a ball was shown to satisfy this property (for \ae driving Brownian motion) in \cite[Theorem 11]{huang2021tight}.
This viewpoint also suggests certain algorithms with explicit links to the ultrametricity of low-temperature Gibbs measures; in particular they seem quite different from Langevin dynamics.

To our knowledge, the only rigorous lower bounds for the energy attained by \eqref{eq:langevin-dynamics} come from \cite{ben2020bounding}, which developed differential inequalities for the dynamical system $U_N(t)=(H_N(\bx_t)/N,\|\nabla H_N(\bx_t)\|^2/N)$ which are valid uniformly over $\bx_t\in\cS_N$ in the large $N$ limit. Their estimates were explicit but non-sharp, see Figure $3$ therein. This approach was later employed in \cite{ben2020algorithmic,ben2021online,ben2022high} to study related statistical problems.
For upper bounds, \cite{gamarnik2020optimization} used a different stability property to prove that $H_N(\bx_T)/N$ is bounded away from the ground state energy by a $\beta$-independent constant with non-negligible probability when $p\geq 4$ is even and $T$ is constant.

\subsection{Our Results}
\label{subsec:results}

Theorems~\ref{thm:mainLB} and \ref{thm:mainUB} below are the main results of this paper. Together they confirm the prediction of \cite{cugliandolo1993analytical}: low temperature Langevin dynamics reaches energy exactly $E_{\infty}$ in the limit of large constant $\beta$ and $T$.

\paragraph{{Statement of the Lower Bound}}

Our lower bound Theorem~\ref{thm:mainLB} is proved by analyzing the dynamics on short $O(1/\beta)$ time-scales from arbitrary $\bx_{\tau}\in\cS_N$ via the local behavior of $H_N$, in line with the physical intuition from \cite{cugliandolo1993analytical,biroli1999dynamical}. 
It holds on activated time-scales $e^{cN}$ for small $c$, even when the initial condition $\bx_0$ is disorder-dependent.
Here and throughout, we write $c$ for a constant which is arbitrarily small depending on all other constants except $N$. It is most often used to indicate that probabilities are at most $e^{-cN}$ or at least $1-e^{-cN}$.

\begin{theorem}
\label{thm:mainLB}
For any $p\geq 2$ and $\eta>0$, for some $T_0=T_0(p,\eta)$ and all sufficiently large $\beta\geq \beta_0(p,\eta)$, for $N$ sufficiently large and any (possibly $H_N$-dependent) $\bx_0\in\cS_N$:
    \begin{equation}
    \label{eq:mainLB}
    \bbP\lt[\inf_{t\in [T_0,T_0+e^{cN}]}
    H_N(\bx_t)/N
    \geq
    E_{\infty}-\eta
    \rt]
    \geq 1-e^{-cN}
    \end{equation}
\end{theorem}

In fact our proof of Theorem~\ref{thm:mainLB} uses very little about the Hamiltonian $H_N$. We show that in great generality, Langevin dynamics reaches and stays above the lowest energy of any $\eps$-approximate local maximum as defined now.

\begin{definition}
\label{def:approx-crit}
    $\bx\in\cS_N$ is an \textbf{$\eps$-approximate critical point} for $H_N$ if $\|\nabla_{\sph} H_N(\bx)\|\leq \eps\sqrt{N}$. 
    If additionally $\lambda_{\lfloor \eps N\rfloor}\big(\nabla_{\sph}^2 H_N(\bx)\big)\leq \eps$, then $\bx\in\cS_N$ is an \textbf{$\eps$-approximate local maximum}. 
    If only the first property holds, $\bx$ is an \textbf{$\eps$-approximate saddle point}.
\end{definition}

\begin{definition}
    For a smooth function $H_N:\cS_N\to\bbR^N$, let $NE_*^{(\eps)}(H_N)$ be the minimal energy of any $\eps$-approximate local maximum $\bx\in\cS_N$.
\end{definition}

{
Theorem~\ref{thm:mainLB} is a direct consequence of the much more general Theorem~\ref{thm:mainLB-general} below.
Write $2\cB_N=\{\bx\in\bbR^N~:~\|\bx\|\leq 2\sqrt{N}\}$.
Given a twice-differentiable function $H_N:2\cB_N\to\bbR^N$, we define the rescaled Sobolev norm $\|H_N\|_{W^{3,\infty}_N}$ as:
\begin{equation}
\label{eq:sobolev-norm}
    \sup_{\bx\in 2\cB_N}
    \max\Big(
    \frac{|H_N(\bx)|}{N}, \frac{\|\nabla H_N(\bx)\|}{\sqrt{N}},\|\nabla^2 H_N(\bx)\|_{\op}, \|\nabla^3 H_N(\bx)\|_{\op} \sqrt{N}
    \Big).
\end{equation}
The operator norms $\|\cdot\|_{\op}$ we use are defined later in \eqref{eq:operator-norm} and are also known as injective tensor norms.
In the setting of \cite[Section 4]{ben2020bounding}, \eqref{eq:sobolev-norm} is equivalent to the norm for the space $\cG^3$.
We say $H_N:2\cB_N\to\bbR^N$ is \textbf{$C$-bounded} if $\|H_N\|_{W^{3,\infty}_N}\leq C$. 
We say $H_N:\cS_N\to\bbR$ is $C$-bounded if it is given by restricting some $C$-bounded function on $2\cB_N$ to $\cS_N$. 
}

\begin{theorem}
\label{thm:mainLB-general}
Let $H_N:\cS_N\to \bbR$ be any $C$-bounded function.
Then for $T_0=T_0(p,\eps,C)$ and all sufficiently large $\beta\geq \beta_0(p,\eps,C)$, for $N$ sufficiently large and any $\bx_0\in\cS_N$:
    \begin{equation}
    \label{eq:mainLB}
    \bbP\lt[\inf_{t\in [T_0,T_0+e^{cN}]}
    H_N(\bx_t)/N
    \geq
    E_*^{(\eps)}(H_N)-\eps
    \rt]
    \geq 1-e^{-cN}
    \end{equation}
    {
    Here in the left-hand probability, the only source of randomness comes from the driving Brownian motion in \eqref{eq:langevin-dynamics} (since $H_N$ and $\bx_0$ are arbitrary).}
\end{theorem}

Proposition~\ref{prop:gradients-bounded} shows that $p$-spin Hamiltonians are $C$-bounded with probability $1-e^{-cN}$ for $C$ depending only on $p$.
Hence Theorem~\ref{thm:mainLB} follows immediately from Theorem~\ref{thm:mainLB-general} and the following easy result proved in Subsection~\ref{subsec:prelim}.

\begin{lemma}
\label{lem:approx-max}
    For any $E<E_{\infty}$, there exists $\eps$ such that for $N$ large,
    \[
    \bbP[E_*^{(\eps)}(H_N)\geq E]\geq 1-e^{-cN}.
    \]
\end{lemma}

We make some remarks on Theorem~\ref{thm:mainLB}.

\begin{remark}
\label{rem:late-times}
Since $\bx_0$ may depend on $H_N$ in Theorem~\ref{thm:mainLB}, the interval $[T_0,T_0+e^{cN}]$ can be replaced by $[T(N),T(N)+e^{cN}]$ for $T(N)\geq T_0$ growing arbitrarily with $N$ --- simply treat $\bx_{T(N)-T_0}$ as $\bx_0$.
\end{remark}

\begin{remark}
    Since $T_0$ does not depend on $\beta$ in Theorem~\ref{thm:mainLB}, it is natural to try interchanging the $\beta\to\infty$ limit with $N\to\infty$ to obtain a corresponding result for gradient flow. However running low-temperature Langevin dynamics for time $T$ is comparable to gradient flow for time $\beta T$. Because this diverges with $\beta$, we were unable to show that gradient flow reaches energy $E_{\infty}-\eta$ in time $O_{\eta}(1)$.
    Relatedly, although we roughly show that $H_N(\bx_t)$ increases until reaching energy $E_{\infty}-o_{\beta}(1)$, this increase occurs along a discrete sequence of times and does not imply monotonicity of $t\mapsto \plim\limits_{N\to\infty} H_N(\bx_t)/N$ for fixed $\beta$ (where $\plim$ denotes limit in probability). 
\end{remark}

\begin{remark}
    Theorem~\ref{thm:mainLB-general} is related to recent work in non-convex optimization.
    Perhaps the most notable of these is \cite{jin2021nonconvex}, which analyzes similar dynamics in discrete-time and also proves that approximate local maxima are reached. 
    Aside from holding uniformly in time with exponentially high probability, a significant difference is that Theorem~\ref{thm:mainLB-general} is completely dimension-free.
    The main algorithm of \cite{jin2021nonconvex} corresponds to Langevin dynamics with both $T$ and $\beta$ at least $(\log N)^5$. (See Appendix B.1 therein, where $\beta^{-1}$ corresponds to $\frac{1}{\eta r^2}$. The number of iterations can be read off from the fraction $\mathscr T/\mathscr F$ at the start of the proof of Theorem 4.4, and each iteration corresponds to $\eta$ units of continuous time.)
    \cite{zhang2017hitting} gave a different analysis in which both $\beta$ and $T$ grow polynomially with $N$.
    We note that these works aim to reach points with \textbf{all} Hessian eigenvalues less than $\eps$, which likely makes at least logarithmic dimension-dependence necessary (see \cite{simchowitz2018tight} for a discrete-time result in this direction).
\end{remark}

\paragraph{On the Proof of Theorem~\ref{thm:mainLB-general}}

To prove Theorem~\ref{thm:mainLB-general} we proceed as follows; here and below, we treat $C$ as a constant which may vary from line to line but should depend only on $p$. First, it is not hard to argue that Langevin dynamics at low temperature gains energy unless the gradient norm $\|\nabla_{\sph} H_N(\bx_t)\|\leq C\beta^{-1/2}\sqrt{N}$ is small.
In this case for $\beta$ large enough that $C\beta^{-1/2}\sqrt{N}\leq\eps\sqrt{N}$, the definition of $E_*^{(\eps)}(H_N)$ implies that either $H_N(\bx_t)\geq E_*^{(\eps)}(H_N)$ already holds, or $\bx_t$ is an $\eps$-approximate saddle point.
The main thrust of our argument is to prove that an energy gain occurs near approximate saddle points obeying $\|\nabla_{\sph} H_N(\bx_t)\|\leq C\beta^{-1/2}\sqrt{N}$, for $\beta$ sufficiently large depending on $\eps$.

To do so, we use a local parametrization for $\cS_N$ to map the dynamics to a flat Euclidean space.
The crucial step is then to use a quadratic Taylor expansion for $H_N$ to
approximate the flattened dynamics by an $N-1$ dimensional Ornstein--Uhlenbeck process.
This approximating process admits a simple direct analysis, and we find that after time $\oC/\beta$ for sufficiently large but $\beta$-independent $\oC(\eps)$, it equilibrates on negative eigenmodes while gaining exponential-in-$\oC$ energy in positive eigenmodes. 
This leads to a net energy gain of at least $\beta^{-1}N$, while the energy discrepancy from the original Langevin dynamics turns out to be of lower order $O(\beta^{-3/2}N)$.

We note that the choice of time-scale on which we show an energy gain is crucial and somewhat delicate. Firstly as observed in \cite{ben2020bounding}, the process $H_N(\bx_t)$ has negative drift whenever $\bx_t$ is a critical point and $H_N(\bx_t)>0$. This was a key source of difficulty for their approach and makes it clear that studying positive time-scales is necessary for sharp results.
In fact our Ornstein--Uhlenbeck idealization of saddle point dynamics requires time $\Omega(1/\beta)$ for the potentially small fraction of positive eigenvalues to outweigh the negative ones.
In the opposite direction, its energy gain is exponential in $\oC$, so on time-scales larger than $O\lt(\frac{\log \beta}{\beta}\rt)$ it ceases to be a good proxy for $\bx_t$.

\paragraph{Statement of the Upper Bound}

Turning next to the upper bound of Theorem~\ref{thm:mainUB}, as previously mentioned it was shown in \cite{huang2021tight,huang2023algorithmic} that no dimension-free Lipschitz algorithm reaches energy above $E_{\infty}$ (see Proposition~\ref{prop:BOGP} below).
The approach we take for the upper bound is to approximate the spherical Langevin dynamics by such a Lipschitz algorithm so that these results can be applied. This implies that for any $T,\eta>0$:
\begin{equation}
\label{eq:original-UB-intro}
    \bbP\big[
    \sup_{t\in [0,T]} H_N(\bx_t)/N\leq E_{\infty}+\eta
    \big]
    \geq 1-e^{-cN}.
\end{equation}
This bound is improved to $E_{\infty}-\delta(\beta)$ using a result of \cite{ben2020bounding} which ensures that $N^{-1/2}\|\nabla_{\sph} H_N(\bx_t)\|$ remains above a positive (but $\beta$-dependent) constant with high probability. The result is as follows.

\begin{theorem}
\label{thm:mainUB}
    \revedit{For any $p\geq 2$ and $\beta_0\in [0,\infty)$, there exists $\delta(p,\beta_0)>0$ such that the following holds for all $\beta\leq\beta_0$}. For any $T\geq 0$, for $N$ sufficiently large and with $\bx_0\in\cS_N$ independent of $H_N$:
    \begin{equation}
    \label{eq:mainUB}
    \bbP\lt[\sup_{t\in [0,T]}
    H_N(\bx_t)/N
    \leq
    E_{\infty}-\delta
    \rt]
    \geq 1-e^{-cN}.
    \end{equation}
\end{theorem}

To approximate the Langevin dynamics by a Lipschitz algorithm and hence prove \eqref{eq:original-UB-intro}, we pass through the soft spherical analog studied in \cite{ben2006cugliandolo,dembo2007limiting}.
As a byproduct, we obtain in Subsection~\ref{subsec:CK} the validity of the Cugliandolo--Kurchan equations for the process \eqref{eq:langevin-dynamics}.
Indeed this was shown in \cite{ben2006cugliandolo} for the \emph{soft} spherical dynamics, and our approximations suffice to transfer their results.

Finally because our bounds match as $\beta\to\infty$, we also find that $\|\nabla_{\sph} H_N(\bx_t)\|$ is small at large constant times.

\begin{corollary}  
\label{cor:grad-small}
For any $p\geq 2$ and positive $\eta$, for all $\beta\geq \beta_0(p,\eta)$ there exists $T_0,c>0$ such that for any $T\geq T_0$, for $N$ sufficiently large and with $\bx_0\in\cS_N$ independent of $H_N$,
    \begin{equation}
    \label{eq:grad-small}
    \bbP\lt[
    \sup_{t\in [T_0,T]}\|\nabla_{\sph}H_N(\bx_t)\|\leq \eta\sqrt{N}
    \rt]
    \geq 1-e^{-cN}.
    \end{equation}
\end{corollary}

\paragraph{On Mixed $p$-Spin Models}

More generally than \eqref{eq:def-hamiltonian}, one can consider the mixed $p$-spin models where
\[
H_N(\bx)=\sum_{p=2}^P \gamma_p N^{-(p-1)/2}
\lt\la\bG_N^{(p)},\bx^{\otimes p}\rt\ra
\]
is an inhomogeneous polynomial.
The situation is expected to be more complicated here, see e.g. \cite[Section 4.4]{cugliandolo2023recent}.
In short, our main proof techniques generalize to mixed models with no changes, but do not give matching thresholds.
Theorem~\ref{thm:mainUB} will feature the algorithmic threshold 
$\ALG(\xi)=\int_0^1 \sqrt{\xi''(t)}~\de t$ for Lipschitz algorithms, where $\xi(t)=\sum_{p=2}^P \gamma_p^2 t^p$ describes the covariance structure of $H_N$ (see \cite{subag2018following,huang2021tight,huang2023algorithmic}).
Meanwhile Theorem~\ref{thm:mainLB} will continue to feature $E_*^{(\eps)}(H_N)$ via the general Theorem~\ref{thm:mainLB-general}. In fact by \cite[Section 7]{kac-rice-in-progress}, this latter threshold is lower bounded by $E_{\infty}^- -o_{\eps}(1)$ with exponentially high probability, where $E_{\infty}^-$ is a Kac--Rice threshold from \cite{auffinger2013complexity}.

The basic reason for this mismatch is that the bulk spectrum of $\nabla_{\sph}^2 H_N(\bx)$ is determined by the radial derivative $\nabla_{\rd} H_N(\bx)$, which is a constant multiple of $H_N(\bx)$ only when $H_N$ is homogeneous (see the proof of Lemma~\ref{lem:approx-max}).
Thus the approximate critical points $\bx$ with bulk spectral edge $\lambda_{\lfloor \eps N\rfloor}(\nabla_{\sph}^2 H_N(\bx))\approx 0$ may occupy a range of energy levels, so it is unclear what threshold energy to expect even heuristically.
In fact, recent simulations in \cite{folena2021gradient} suggest the following in mixed models. If the dynamics are started at high temperature $\beta_0^{-1}$ before later switching to zero temperature gradient flow, the asymptotic energy achieved depends nontrivially on $\beta_0$.
Modulo the difference between zero and low temperature, this does not happen in pure models: all our results continue to hold if $\bx_0$ is obtained by a preliminary phase of Langevin dynamics (at arbitrary bounded temperature, for bounded time).

Finally for the \emph{pure multi-species} spherical spin glasses studied in \cite{kivimae2021ground,subag2021tap2}, Theorems~\ref{thm:mainLB} and \ref{thm:mainUB} extend with trivial changes and match at a generalized $E_{\infty}$ threshold defined in \cite[Section 1.3.3]{huang2023algorithmic} (see also \cite{mckenna2021complexity}).

\subsection{Notations and Technical Preliminaries}
\label{subsec:prelim}

Here we define further notation and state some useful results.
First, we define $\vzero=(0,0,\dots,0)\in \bbR^N$.
Recall that $\cS_N = \lt\{\brho \in \bbR^N : \sum_{i=1}^N \brho_i^2 = N\rt\}$ and let $\cB_N= \lt\{\brho \in \bbR^N : \sum_{i=1}^N \brho_i^2 \leq N\rt\}$ denote the corresponding ball.
Let $S_{\bx}^{\perp}$ denote the orthogonal complement of a non-zero vector $\bx\in\bbR^N\backslash \{\vzero\}$, and $P_{\bx}^{\perp}$ the orthogonal projection matrix onto $S_{\bx}^{\perp}$.
Let $\sH_N$ denote the space of Hamiltonians $H_N$, identified with their underlying $p$-tensor $\bG_N^{(p)}\in(\bbR^N)^{\otimes p}$ and metrized by entry-wise $L^2$ distance via
\begin{equation*}
\label{eq:tensor-L2}
\|H_{N}-H_{N}'\|^2
=
\sum_{i_1,\dots,i_p=1}^N
|g_{i_1,\dots,i_p}
-
g'_{i_1,\dots,i_p}|^2.
\end{equation*}
In particular an $L$-Lipschitz (optimization) algorithm is simply an $L$-Lipschitz function $\cA:\sH_N\mapsto \cB_N$. 
The significance of this class of functions for us is the following result ensuring the energy attained by any $L$-Lipschitz optimization algorithm is at most $E_{\infty}+o_N(1)$, for any fixed $L$ as $N\to\infty$. 
This was first shown by the author and Huang in \cite{huang2021tight} for even $p$, using an extension of the overlap gap property \cite{gamarnik2014limits,gamarnik2017performance,gamarnik2021survey} which was inspired by ultrametricity of the Gibbs measures --- an overview can be found at the end of \cite{auffinger2022optimization}.
Our subsequent work \cite{huang2023algorithmic} on multi-species models is cited below because it allows $p$ to be odd (by avoiding the interpolation method).

\begin{proposition}[{\cite[Theorem 1 and Corollary 1.8]{huang2023algorithmic}}]
\label{prop:BOGP}
    Fix $p,L,\eta$. For $N$ sufficiently large, any $L$-Lipschitz $\cA_N:\sH_N\to\cB_N$ satisfies
    \[
    \bbP\lt[H_N(\cA(H_N))/N\geq E_{\infty}+\eta\rt]
    \leq e^{-cN}.
    \]
\end{proposition}

In addition to the standard gradients and Hessians for functions defined on $\bbR^N$, for $\bx\in\cS_N$ (or occasionally other $\bx\in\bbR^N\backslash \{\vzero\}$) we define the radial and spherical derivative and the Riemannian Hessian:
\begin{align}
\notag
    \nabla_{\rd}H_N(\bx)
    &=
    \frac{1}{N}
    \frac{\de}{\de t}
    H_N(t\bx)\Big|_{t=1}
    \\
\label{eq:spherical-derivative}
    \nabla_{\sph}H_N(\bx)
    &=
    P_{\bx}^{\perp}
    \nabla H_N(\bx)
    \\
\notag
    \nabla_{\sph}^2 H_N(\bx)
    &=
    P_{\bx}^{\perp}
    \big(\nabla^2 H_N(\bx)-\nabla_{\rd}H_N(\bx)\cdot I_N\big)
    P_{\bx}^{\perp}.
\end{align}

Operator norms of $N\times N$ matrices and $N\times N\times N$ tensors will be denoted by
\begin{equation}
\label{eq:operator-norm}
\begin{aligned}
\|\bM_N\|_{\op}
&=
\sup_{\bx,\by\in\bbR^N\backslash\{\vec 0\}} \frac{\la \bM_N,\bx\otimes \by\ra}{\|\bx\|\cdot\|\by\|},
\\
\|\bT_N\|_{\op}
&=
\sup_{\bx,\by,\bz\in\bbR^N\backslash\{\vec 0\}} \frac{\la \bT_N,\bx\otimes \by\otimes\bz\ra}{\|\bx\|\cdot\|\by\|\cdot\|\bz\|}.
\end{aligned}
\end{equation}
Recall the definition of $C$-boundedness just before Theorem~\ref{thm:mainLB-general}.
The proof of the next proposition is routine and omitted (recall $C$ may change from line to line).

\begin{proposition}
\label{prop:gradients-bounded-0}
Suppose $H_N$ is $C$-bounded. Then:
\begin{enumerate}[label=(\alph*)]
        \item 
        \label{it:grad-bounded}
        For all $\bx, \by\in 2\cB_N$:
        \begin{align}
            \label{eq:gradient-bounded}
            \norm{\nabla^k H_N(\bx)}_{\op}
            &\le 
            CN^{1-\frac{k}{2}},\quad\forall~0\leq k\leq 3 \\
            \notag
            \norm{\nabla^k H_N(\bx) - \nabla^k H_N(\by)}_{\op}
            &\le 
            CN^{\frac{1-k}{2}}\norm{\bx - \by},
            \quad\forall~0\leq k\leq 2.
        \end{align}
        \item 
        \label{it:spherical-grad-bounded}
        For all $\bx,\by\in 2\cB_N\backslash \frac{1}{2}\cB_N$,
        \begin{align}
            \notag
            |\nabla_{\rd} H_N(\bx)|
            &\le 
            C,
            \\
            \notag
            |\nabla_{\rd} H_N(\bx) - \nabla_{\rd} H_N(\by)|
            &\le 
            C\norm{\bx - \by},
            \\
            \notag
            \norm{\nabla_{\sph} H_N(\bx)}
            &\le 
            C, \\
            \notag
            \norm{\nabla_{\sph} H_N(\bx) - \nabla_{\sph} H_N(\by)}
            &\le 
            C\norm{\bx - \by},
            \\
            \notag
            \norm{\nabla_{\sph}^2 H_N(\bx)}_{\op}
            &\le 
            C.
        \end{align}
    \end{enumerate}
\end{proposition}

\begin{proposition}
\label{prop:gradients-bounded}
    For any $p$ there exist constants $C,c>0$ such that $H_N$ defined by \eqref{eq:def-hamiltonian} is $C$-bounded on $2\cB_N$ with probability $1-e^{-cN}$, for large $N$. Moreover the set $K_N\subseteq \sH_N$ of $C$-bounded $H_N$ is convex.
\end{proposition}

\begin{proof}
    The first claim is \cite[Proposition 2.3]{huang2021tight}, while the convexity is clear.
\end{proof}

\revedit{
We note that convexity of $K_N$ will be used only to apply the Kirzsbraun extension theorem within the proof of Lemma~\ref{lem:soft-langevin-lipschitz}.
}

\begin{proof}[Proof of Lemma~\ref{lem:approx-max}]
    First, the \emph{Euclidean} tangential Hessian 
    $P_{\bx}^{\perp}\nabla^2 H_N(\bx)P_{\bx}^{\perp}$
    satisfies
    \begin{equation}
    \label{eq:euclidean-hessian}
    \lambda_{\lfloor \eps N\rfloor}\lt(P_{\bx}^{\perp}\nabla^2 H_N(\bx)P_{\bx}^{\perp}\rt)
    \geq 
    2\sqrt{p(p-1)}-\eta,\quad\forall~\bx\in\cS_N
    \end{equation}
    with probability $1-e^{-cN}$ for some $\eta>0$ tending to $0$ as $\eps\to 0$.
    This holds because of the $N^2$ speed in the large deviation principle for the GOE bulk spectrum, which lets one union bound over an $\eta$-net on $\cS_N$; see e.g.\ \cite[Lemma 3]{subag2018following}.

    Next since $H_N(\bx)$ is a homogeneous degree $p$ polynomial of $\bx$, we \revedit{always have}:
    \begin{equation}
    \label{eq:nabla-rd}
    \nabla_{\rd}H_N(\bx)
    =
    pH_N(\bx)/N,\quad\forall~\bx\in\cS_N.
    \end{equation}
    Then \eqref{eq:euclidean-hessian} implies that with probability at least $1-e^{-cN}$, for all $\bx\in\cS_N$ such that $H_N(\bx)\leq EN$:
    \begin{align*}
    \lambda_{\lfloor \eps N\rfloor}(\nabla_{\sph}^2 H_N(\bx))
    &=
    \lambda_{\lfloor \eps N\rfloor}\lt(P_{\bx}^{\perp}\nabla^2 H_N(\bx)P_{\bx}^{\perp}\rt)
    -
    pH_N(\bx)
    \\
    &\geq 
    2\sqrt{p(p-1)}-\eta
    -
    pE
    \\
    &=
    p(E_{\infty}-E)-\eta
    \\
    &\geq 
    \eps.
    \end{align*}
    Here the last step holds for small enough $\eps$ since $\eta$ tends to $0$ with $\eps$. This concludes the proof.
\end{proof}

The next proposition will be used to argue that $H_N$ typically increases while the gradient is reasonably large.
In fact \cite[Theorem 4.5]{ben2020bounding} shows 
\[
\sup_{\bx\in\cS_N}\lt|\frac{\Tr\,\nabla_{\sph}^2 H_N(\bx)-pH_N(\bx)}{N}\rt|\leq \delta
\]
with probability $1-e^{-cN^2}$ for pure $p$-spin models, but we will not need this. \revedit{(Note also that $\Tr\,\nabla_{\sph}^2(\cdot)$ is just the Riemannian Laplacian on the sphere.)}

\begin{proposition}
\label{prop:d-HNM}
$H_N(\bx_t)$ evolves according to the SDE:
\begin{equation}
\label{eq:d-HN}
    \revedit{\de H_N(\bx_t)
    =
    \Big(
    \beta \|\nabla_{\sph} H_N(\bx_t)\|^2
    -\Tr\,\nabla_{\sph}^2 H_N(\bx_t)
    \Big)\de t
    +
    \|\nabla_{\sph} H_N(\bx_t)\| \de B_t}.
\end{equation}
If $H_N$ is $C$-bounded then $\sup_{\bx\in\cS_N}|\Tr\,\nabla_{\sph}^2 H_N(\bx)|\leq CN$.
\end{proposition}

\begin{proof}
    \Ito's formula immediately gives \eqref{eq:d-HN}. The latter bound follows from the definition of $C$-boundedness.
\end{proof}

We now give some remarks on constants (recall also the discussion of $c$ at the start of Subsection~\ref{subsec:results}). Throughout the main proofs, we will use $C$ to denote a constant depending only on $p$ (in particular \emph{not} on $\beta$) which may vary from line to line. In particular, we may assume $H_N$ is $C$-bounded and then use other values of $C$ in the ensuing proof.
In Section~\ref{sec:LB} we will consider a large constant $\oC$ which depends also on $E<E_{\infty}$ but is still independent of $\beta$. 

\subsubsection{Estimates for $1$-Dimensional Diffusions}
\label{subsubsec:1-d-estimates}


We next state two elementary estimates for $1$-dimensional diffusion processes. They essentially ensure that differential inequalities are not ruined by a vanishing amount of noise.
These will be used repeatedly to estimate the growth of various approximation errors, for example in \eqref{eq:d-HN}.
We note that $\tau'$ is used throughout the paper to denote assorted stopping times. The use of $\tau$ on its own is reserved for $\bx_{\tau}$ as explained at the start of Section~\ref{sec:LB}.

\revedit{
Here and throughout the rest of the paper, we will use the compact notation $X\in\cI_{\tau'}(b,\sigma)$ to indicate that $X_t$ weakly solves the scalar SDE
\begin{equation}
\label{eq:cI-SDE}
\de X_t=b_t~\de t+ \sigma_t ~\de B_t
\end{equation}
for coefficients $b_t,\sigma_t$ (progressively measurable with respect to the filtration) until a stopping time $\tau'$. 
Thus $\cI_{\tau'}(b,\sigma)$ denotes the set of weak solutions to this equation.
We write $\cI(b,\sigma)$ for $\cI_{\infty}(b,\sigma)$.
}

\begin{proposition}
\label{prop:almost-monotone}
    \revedit{Fix constants $C_1,C_2,C_3,\eps,s>0$ and suppose the $1$-dimensional process $X_t$, initialized with $X_0=0$, satisfies $X\in \cI_{\tau'}(b,\sigma)$ for $b_t\in [C_1,C_2]$ and $\sigma_t\leq C_3/\sqrt{N}$. Then with probability $1-e^{-cN}$, either $\tau'\leq s$ or 
    \[
    \min_{t\in [0,s]} 
    (X_t-C_1t) \geq -\eps.
    \]
    }
\end{proposition}

\begin{proof}
    By the Dubins--Schwarz representation of stochastic integrals as time-changes of Brownian motion (see e.g. \cite[Section V.1]{revuz2013continuous}), we find that
    \[
    \bbP\lt[
    \min_{t\in [0,s\wedge\tau']}
    \int_0^t \sigma_t~\de B_t
    \geq -\eps
    \rt]
    \geq 1-e^{-cN}
    .\]
    This implies the claim.
\end{proof}

\begin{proposition}
\label{prop:self-bounding}
    \revedit{
    Suppose $A_t\in \cI_{\tau'}(b_t,\sigma_t)$ is a non-negative scalar diffusion where, for constants $C_1,C_2,C_3,C_4\geq 0$ with $C_2,C_3>0$ strictly positive, the coefficients satisfy for all $t\leq \tau'$:
    \begin{align*}
    b_t
    &\leq 
    C_1 A_t+ C_2\sqrt{A_t} + C_3,
    \\
    |\sigma_t|
    &\leq 
    C_4 N^{-1/2}
    \end{align*}
    Then for any fixed $s>0$, with probability $1-e^{-cN}$:
    \[
    \sup_{t\in [0,s\wedge \tau']}
    A_t
    \leq 
    e^{C_1 s}\lt(\lt(2C_2s + \sqrt{A_0+C_3^2/C_2^2}\rt)^2-C_3^2/C_2^2\rt).
    \]
    }
\end{proposition}

\begin{proof}
    Let $\wt A_t=e^{-C_1 t} A_t $.
    Then \revedit{by \Ito's formula},
    \[
    \de \wt A_t=\wt b_t~\de t+\wt\sigma_t~\de B_t
    \]
    with 
    \begin{align*}
    |\wt\sigma_t|
    &
    \revedit{=e^{-C_1 t}|\sigma_t|}
    \leq 
    C_4 N^{-1/2},
    \\
    \wt b_t 
    &\leq
    e^{-C_1 t}
    \lt(
    b_t - C_1 A_t
    \rt)
    \\
    &\leq 
    C_2\sqrt{\wt A_t}+C_3
    \leq 
    2C_2 \sqrt{\wt A_t + C_3^2/C_2^2}.
    \end{align*}
    From this, letting $\wh A_t=\sqrt{\wt A_t+C_3^2/C_2^2}$ we find
    \[
    \de \wh A_t=\wh b_t~\de t+\wh\sigma_t~\de B_t
    \]
    where since $\wh A_t\geq C_3/C_2$:
    \begin{align*}
    |\wh\sigma_t|
    &\leq 
    C_4\sqrt{\frac{C_2 N}{C_3}},
    \\
    \wh b_t 
    &\leq 
    \frac{\wt b_t}{2\wh A_t^{1/2}}
    \leq 
    C_2.
    \end{align*}
    Using Proposition~\ref{prop:almost-monotone}, we find that with probability $1-e^{-cN}$,
    \[
    \sup_{t\leq  s\wedge\tau'} \wh A_t \leq 2C_2 s+ \wh A_0 
    = 
    2C_2s + \sqrt{A_0+C_3^2/C_2^2}.
    \]
    On this event, we obtain
    \begin{align*}
    \sup_{t\leq  s\wedge \tau'}
    A_t
    &\leq
    e^{C_1 s}\lt(\lt(2C_2s + \sqrt{A_0+C_3^2/C_2^2}\rt)^2-C_3^2/C_2^2\rt).
    \qedhere
    \end{align*}
\end{proof}

\begin{remark}
\label{rem:existence-of-solutions}
    \revedit{In the statements above, we do \textbf{not} assume the coefficients $b_t,\sigma_t$ are Lipschitz, nor that the SDEs admit unique strong solutions.
    For example in computing using \Ito's formula or using the Dubins--Schwarz representation, it makes no difference whether the driving Brownian motion must be defined on an extended probability space. 
    In our applications of the results above, $X_t,A_t$ will be obtained from $N$-dimensional SDEs with Lipschitz coefficients using \Ito's formula. Hence existence of solutions will be automatic.
    We also note that in the resulting SDE \eqref{eq:cI-SDE}, the coefficients $b_t,\sigma_t$ might not be adapted to the filtration generated by $B_t$, but will be progressively measurable only with respect to the natural filtration of the original $\bbR^N$-valued processes. (Due to this last point, it does not even make sense to say they are Lipschitz since they depend on more than just $(t,X_t)$.)
    }
\end{remark}

\revedit{
The estimates above will be applied to show that pairs of $N$-dimensional diffusions remain close together, by explicit consideration of their squared distance. 
We next describe explicitly the scalar SDE produced by such arguments.

\begin{proposition}
\label{prop:bt-sigma-bound}
    Let $\wh\by_t,\wt\by_t$ be $N$-dimensional processes (possibly with different initializations $\wh\by_0,\wt\by_0$) strongly solving the SDEs
    \begin{align*}
    \de\wh\by_t &= \wh \bb_t ~\de t + \wh\bsig_t~\de \bB_t
    ,
    \\
    \de\wt\by_t &= \wt \bb_t ~\de t + \wt\bsig_t~\de \bB_t,
    \end{align*}
    where all coefficients are spatially Lipschitz. (Here $\wh \bb_t,\wt \bb_t$ are $N$-dimensional drift coefficients while $\wh\bsig_t,\wt\bsig_t$ are $N\times N$ matrices.)
    Let $w_t=\|\wh\by_t-\wt\by_t\|^2/N$. 
    Then (with $\|\cdot\|_F$ the Frobenius matrix norm) we have $w\in \cI(b,\sigma)$ for:
    \begin{equation}
    \label{eq:w-t-general-formula}
    \begin{aligned}
    b_t
    &=
    \frac{2}{N}\la \wh \bb_t-\wt \bb_t, \wh\by_t-\wt \by_t\ra
    +
    \frac{1}{N}\|\wh\bsig_t - \wt\bsig_t\|_F^2
    ,
    \\
    \sigma_t
    &= 
    \frac{2}{N}\|(\wh\bsig_t - \wt\bsig_t)(\wh\by_t-\wt \by_t)\|
    \leq 
    2\|\wh\bsig_t - \wt\bsig_t\|_{\op} \sqrt{w_t/N}.
    \end{aligned}
    \end{equation}
\end{proposition}

\begin{proof}
    Both statements are immediate from \Ito's formula.
\end{proof}
}

\section{Lower Bound: Proof of Theorem~\ref{thm:mainLB}}
\label{sec:LB}

{The main goal of this section is to prove Theorem~\ref{thm:mainLB-general}. As explained in the introductory discussion, this implies Theorem~\ref{thm:mainLB} thanks to Lemma~\ref{lem:approx-max} and Proposition~\ref{prop:gradients-bounded}. 
Thus as in the statement of Theorem~\ref{thm:mainLB-general}, $H_N$ will always be a determininistic $C$-bounded function in this section (which will be clarified throughout).
}
In Subsection~\ref{subsec:energy-gain} we show that $\bx_t$ gains energy while $\|\nabla_{\sph}H_N(\bx_t)\|\geq C\beta^{-1}\sqrt{N}$. In the two following subsections, we develop the crucial Ornstein--Uhlenbeck approximation to the Langevin dynamics near a saddle point. Finally in Subsection~\ref{subsec:escape-saddle} we establish the required energy gain for the Ornstein--Uhlenbeck approximation. The final proof then combines these two cases by exhibiting a sequence $0=\tau_0<\tau_1<\dots<\tau_M$ of stopping times separated by $O(1/\beta)$ along which $H_N(\bx_{\tau_{m}})$ increases whenever $H_N(\bx_{\tau_{m}})<\lt(E_{\infty}-\frac{\eta}{2}\rt)N$.

Throughout this section we will use $\bx_{\tau}\in\cS_N$ to indicate $\bx_t$ stopped at an arbitrary stopping time $\tau$ and study the behavior of $\bx_t$ on small future time intervals such as $t\in [\tau,\tau+C\beta^{-1}]$. 
All these intermediate results hold for arbitrary disorder-dependent initial conditions $\bx_0\in \cS_N$ in place of $\bx_{\tau}$.
However when we eventually implement the main argument described above, $\tau$ will directly range over the sequence $\tau_0,\tau_1,\dots,\tau_M$. 
{
We use the notation $\cF_{\tau}$ for the sigma-field generated by the dynamics until time $\tau$. 
Since $H_N$ is treated as non-random in this section, one may equivalently think that $\cF_{\tau}$ is generated by these dynamics as well as $H_N$ itself.}

\subsection{A Priori Estimates and Energy Gain Away From Approximate Critical Points}
\label{subsec:energy-gain}

\revedit{
We begin with the following useful lemma which controls the short-time movement of $\bx_t$. The first part should be considered the main statement, as the second part is used only in proving Lemma~\ref{lem:apriori-quadratic-SDE-approximation}. In this section we always take $\beta$ large, in which case $\beta+1$ below can be replaced by just $\beta$. However Lemma~\ref{lem:movement-bound} is also used in the proof of Theorem~\ref{thm:mainUB} in the next section to upper bound the energy attained, for both small and large values of $\beta$.
}

\begin{lemma}
\label{lem:movement-bound}
    For any $C$-bounded $H_N$, $\beta>0$, $\gamma\leq C$ and $s>0$ the following holds. Suppose $\|\nabla_{\sph} H_N(\bx_{\tau})\|\leq \gamma\sqrt{N}$. Then for $N$ sufficiently large, with probability $1-e^{-cN}$ \revedit{conditionally on $\cF_{\tau}$}:
    \begin{equation}
    \label{eq:same-conclusion}
    \sup_{t\in [\tau,\tau+s]}
    \|\bx_t-\bx_{\tau}\|/\sqrt{N}
    \leq 
    \revedit{
    4
    e^{C(\beta+1) s}
    \sqrt{(\beta+1)^2\gamma^2 s^2+s}
    }
    .
    \end{equation}
    \revedit{More generally, suppose $\wh\bx_t$ solves the $\bbR^N$-valued SDE
    \[
    \de \wh\bx_t=\bb(\wh\bx_t)~\de t+ \bsig_t\de\bB_t
    \] 
    on $t\in [\tau,\tau']$
    where $\|\bsig_t\|_{\op}\leq 1$, $\|\bb_{\tau}\|\leq \gamma(\beta+1)\sqrt{N}$, $\|\bb_t\|\leq C(\beta+1)\sqrt{N}$ and $\|\bb_{\tau}-\bb_t\|\leq C(\beta+1)\|\wh\bx_t-\wh\bx_{\tau}\|$ hold almost surely for all $t$.
    (Here $\bb_t=\bb(\wh\bx_t)$ is a time-invariant vector-valued function of $\wh\bx_t$.)
    Then the same bound \eqref{eq:same-conclusion} holds uniformly on $t\in [\tau,\tau'\wedge (\tau+s)]$.
    }
\end{lemma}

\begin{proof}
    \revedit{We prove the latter claim, which immediately encompasses the former one, and for simplicity assume $\tau'=\infty$ which makes no difference.
    Let $w_t=\frac{\|\wh\bx_t-\wh\bx_{\tau}\|^2}{N}$. 
    \revedit{
    Proposition~\ref{prop:bt-sigma-bound} shows that $w\in \cI(b,\sigma)$ for:
    \begin{align*}
    b_t 
    &=
    2 N^{-1} 
    \lt\la \bb_t,
    \wh\bx_t-\wh\bx_{\tau}\rt\ra
    +
    \|\bsig_t\|_F^2/N
    ;
    \\
    |\sigma_t| &\leq 10N^{-1/2}.
    \end{align*}
    }
    By assumption on $\bsig$, we find $\|\bsig_t\|_F^2/N\leq 1$.
    Using the Lipschitz condition on $\bb(\cdot)$, we find
    \[
    \|\bb(\wh\bx_t)\|
    \leq 
    (\beta+1)\cdot\big(\gamma\sqrt{N}+C\sqrt{Nw_t}\big).
    \]
    Hence combining, we see that 
    \begin{equation}
    \label{eq:bt-bound-near-approx-crit}
    |b_t|\leq 
    2C(\beta+1) w_t + 2(\beta+1) \gamma\sqrt{w_t} + 1 
    .
    \end{equation}
    Set $\wt\beta=2(\beta+1)$.
    Applying Proposition~\ref{prop:self-bounding} to $w_t$ and noting that $A_0=w_{\tau}=0$, we find that with probability $1-e^{-cN}$:
    \begin{align*}
    \sup_{t\in [\tau,\tau+s]} 
    w_t
    &\leq 
    e^{C\wt\beta s}
    \lt(
    (2\wt\beta\gamma s + \wt\beta^{-1}\gamma^{-1})^2 -\wt\beta^{-2}\gamma^{-2}
    \rt)
    \\
    &= 
    4e^{C\wt\beta s}
    \big(\wt\beta^2\gamma^2 s^2 + s \big)
    .
    \end{align*}
    This completes the proof (recall that $C$ may vary from line to line).
    }
\end{proof}

\begin{proposition}
\label{prop:gradient-improve}
    For $\beta\geq 1$ and $C$ a large constant, suppose $\|\nabla_{\sph}H_N(\bx_{\tau})\|\geq C\beta^{-1/2}\sqrt{N}$ and $H_N$ is $C$-bounded.
    Then with probability $1-e^{-cN}$ \revedit{conditionally on $\cF_{\tau}$}:
    \begin{align}
    \label{eq:low-movement-concrete}
    \sup_{s\in [\tau, \tau+\beta^{-3}]} 
    \|\bx_s-\bx_{\tau}\|
    &\leq 
    C'\beta^{-3/2}\sqrt{N};
    \\
    \label{eq:energy-no-loss}
    \inf_{s\in [\tau, \tau+\beta^{-3}]}
    H_N(\bx_s)
    &\geq 
    H_N(\bx_{\tau}) - C'\beta^{-3/2}N;
    \\
    \label{eq:large-grad-uniform}
    \inf_{s\in [\tau, \tau+\beta^{-3}]}
    \|\nabla_{\sph} H_N(\bx_s)\|
    &>
    C\beta^{-1/2}\sqrt{N}/2;
    \\
    \label{eq:energy-gain-simple}
    H_N(\bx_{\tau+\beta^{-3}})
    &\geq 
    H_N(\bx_{\tau}) + \frac{C\beta^{-3}N}{4}.
    \end{align}
\end{proposition}

\begin{proof}
    First we apply Lemma~\ref{lem:movement-bound} on time-scale $\beta^{-3}$ with $\gamma$ a large constant and use that $H_N$ is $C$-bounded. We find that with probability $1-e^{-cN}$,
    \[
    \sup_{\tau\leq s\leq \tau+C\beta^{-1}} \|\bx_s-\bx_{\tau}\|\leq C'\beta^{-3/2}\sqrt{N}.
    \]
    This directly implies \eqref{eq:low-movement-concrete}, hence also \eqref{eq:energy-no-loss} and \eqref{eq:large-grad-uniform} using again $C$-boundedness of $H_N$.

    Turning to \eqref{eq:energy-gain-simple}, let $\tau'$ be the first time that $\|\nabla_{\sph} H_N(\bx_s)\|\leq C\beta^{-1}\sqrt{N}/2$. In light of \eqref{eq:large-grad-uniform} we restrict to the event $\tau'\geq \tau+\beta^{-3}$. Then writing
    \[
    \de H_N(\bx_t)/N
    =
    b_t~\de t
    +
    \sigma_t~\de B_t,
    \]
    it is easy to see from \eqref{eq:d-HN} that $b_t\geq C/3$ and $|\sigma_t|\leq O_{\beta}(N^{-1/2})$.
    Proposition~\ref{prop:almost-monotone} now yields \eqref{eq:energy-gain-simple}.
\end{proof}

\subsection{Stereographic Projection}

To show an energy gain near approximate saddle points, we first parametrize a neighborhood of $\cS_N$ around $\bx_t$ by Euclidean space. Our strategy will be to approximate the reparametrized dynamics by Langevin dynamics for a quadratic potential function, which can be analyzed explicitly as an Ornstein--Uhlenbeck process.

Given $\bz\in \cS_N$, recall that the stereographic projection map $\Gamma_{\bz}:\cS_N\backslash\{-\bz\}\mapsto S_{\bz}^{\perp}$ \revedit{to the orthogonal complement of $\bz$} is defined so that $(-\bz,\bx,\Gamma_{\bz}(\bx))$ are collinear. In coordinates, if $\bz=(\sqrt{N},0,\dots,0)$ then
\begin{equation}
\label{eq:sterographic-coordinates}
    \Gamma_{\bz}(x_1,x_2,\dots,x_N) 
    = 
    \frac{\sqrt{N}}{x_1+\sqrt{N}}
    \cdot 
    (0,x_2,\dots,x_N).
\end{equation}
It is well-known that $\Gamma_{\bz}$ is conformal, i.e. its derivative $\Gamma_{\bz}'(\bx)$ at any $\bx\in\cS_N$ is a non-negative scalar $h_{\bz}(\bx)$ times an isometry.
To compute this scalar, suppose $\bz=(\sqrt{N},0,\dots,0)$ and $\bx=(\sqrt{N}\cos(\theta),\sqrt{N}\sin(\theta),0,\dots,0)$. Then
\begin{equation}
\label{eq:Gamma-formula}
\Gamma_{\bz}(\bx)=\lt(0,\frac{\sqrt{N}\sin\theta}{1+\cos\theta},0,\dots,0\rt)
\end{equation}
and so
\begin{equation}
\label{eq:scalar-deriv-formula}
    h_{\bz}(\bx)= 
    \lt|
    \frac{\de}{\de \theta} 
    \lt(\frac{\sin\theta}{1+\cos\theta}\rt)
    \rt|
    =
    \frac{1}{1+\cos(\theta)}.
\end{equation}
By spherical symmetry, \eqref{eq:scalar-deriv-formula} holds in general with
\revedit{$\theta=\theta(\bx)\geq 0$} the angle between $\bz$ and $\bx$.
Let $\wh \bx_t=\Gamma_{\bz}(\bx_t)$ be the pushforward of the Langevin dynamics, and $\theta_t=\theta(\bx_t)$.
Define the pushforward Hamiltonian $\wh H_N(\Gamma_{\bz}(\bx))\equiv H_N(\bx)$.

\begin{proposition}
\label{prop:sterographic-langevin}
    The law of the process $\wh\bx_t$ obeys the SDE
    \begin{equation}
    \label{eq:pushforward-Langevin}
    \de \wh \bx_t
    =
    \frac{1}{(1+\cos\theta_t)^2}
    \lt(
    \beta\nabla \wh H_N(\wh\bx_t)
    -
    \frac{(N-3)|\sin\theta_t|}{2N^{1/2}}
    \cdot \frac{\wh\bx_t}{\|\wh\bx_t\|}
    \rt)
    \de t
    +
    \frac{1}{1+\cos\theta_t}~\de \bW_t.
    \end{equation}
    for $\bW_t$ a standard Brownian motion in $S_{\bz}^{\perp}\simeq \bbR^{N-1}$.
\end{proposition}

\begin{proof}
    By spherical symmetry it suffices to consider the case $ \bz=\big(\sqrt{N},0,\dots,0\big)$,
    \[
    \bx_t=\big(\sqrt{N}\cos(\theta_t),\sqrt{N}\sin(\theta_t),0,\dots,0\big),
    \quad
    \wh\bx_t=\lt(0,\frac{\sqrt{N}\sin\theta_t}{1+\cos\theta_t},0,\dots,0\rt)
    \]
    with $\theta_t\in [0,\pi)$.
    Conformality of $\Gamma_{\bz}$ yields the diffusive term.

    For the drift term, the contribution of $\nabla \wh H_N(\wh\bx_t)$ is direct from the chain rule. 
    (Note that the drift itself is scaled by $(1+\cos\theta_t)^{-1}$, while $\|\nabla\wh H_N(\wh\bx_t)\|=(1+\cos\theta_t)\|\nabla H_N(\bx_t)\|$ has the opposite scaling.)
    The remaining \Ito\ drift term is parallel to $\wh\bx_t$ by rotational symmetry, so it suffices to compute its contribution in the second coordinate of $\wh\bx_t$.
    Thus, we let $(\by_s)_{s\geq t}$ be spherical Brownian motion started from $\bx_t$ and $\wh\by_s=\Gamma_{\bz}(\by_s)$, and compute the drift of $\wh\by_s$.

    Note that at time $t$, with $\by_s=(y_1(s),y_2(s),\dots,y_N(s))$ we have 
    \begin{align*}
    \de y_1(t)
    &=
    -\frac{(N-1)\cos(\theta_t)}{2\sqrt{N}} ~\de t
    +
    \sin(\theta_t)~\de B_t,
    \\
    \de y_2(t)
    &=
    -\frac{(N-1)\sin(\theta_t)}{2\sqrt{N}} ~\de t
    -
    \cos(\theta_t)~\de B_t.
    \end{align*}

    From \Ito's formula it follows that the drift of $\wh y_2(s)=\frac{y_2(s)\sqrt{N}}{y_1(s)+\sqrt{N}}$ at time $t$ is
    \begin{align*}
    &-\frac{(N-1)\sin(\theta_t)}{2\sqrt{N}}
    \cdot \frac{\sqrt{N}}{y_1(t)+\sqrt{N}}
    +
    \frac{(N-1)\cos(\theta_t)}{2\sqrt{N}}
    \cdot 
    \frac{y_2(t)\sqrt{N}}{(y_1(t)+\sqrt{N})^2} 
    +
    \frac{\sin(\theta_t)}{\sqrt{N}(1+\cos \theta_t)^2}
    \\
    &=
    -\frac{(N-1)\sin(\theta_t)}{2\sqrt{N}(1+\cos\theta_t)}
    +
    \lt(\frac{(N-1)\cos(\theta_t)}{2}+1\rt)
    \cdot 
    \frac{\sin\theta_t}{\sqrt{N}(1+\cos\theta_t)^2} 
    \\
    &=
    -\frac{(N-3)\sin(\theta_t)}{2\sqrt{N}(1+\cos\theta_t)}
    +
    \frac{(N-3)\cos(\theta_t)\sin(\theta_t)}{2\sqrt{N}(1+\cos\theta_t)^2}
    \\
    &=
    -\frac{(N-3)\sin(\theta_t)}{2\sqrt{N}(1+\cos\theta_t)^2}
    \end{align*}
    where the last term in the first line is the 
    \Ito\ correction 
    \[
    \frac{1}{2}\frac{\de^2}{\de \delta^2}
    \lt(
    \frac{N^{1/2}\sin(\theta_t)-\delta\cos(\theta_t)}
    {1+\cos(\theta_t)+N^{-1/2}\delta\sin(\theta_t)}
    \rt)\Bigg|_{\delta=0}.
    \qedhere
    \]
\end{proof}

\subsection{Quadratic Potential Approximation}
\label{subsec:quadratic-approx}

Fix $\eps>0$ and assume $H_N:\cS_N\to\bbR$ is $C$-bounded and $\beta\geq \beta_0(\eps,C)$ is large.
In light of Proposition~\ref{prop:gradient-improve}, we consider an arbitrary $\bx_{\tau}\in\cS_N$ satisfying 
\begin{align}
\label{eq:approx-crit-start}
    \|\nabla_{\sph}H_N(\bx_{\tau})\|&\leq C\beta^{-1/2}\sqrt{N}.
\end{align} 
We will use $\bz=\bx_{\tau}$ for stereographic projection.
For $\wh\bx\in S_{\bz}^{\perp}$, define the quadratic Taylor expansion:
\begin{equation}
\label{eq:HNQ}
    H_N^{(Q)}(\wh\bx)
    =
    \wh H_N(\vec 0)
    +
    \la\nabla\wh H_N(\vec 0),\wh\bx\ra
    +
    \frac{1}{2}\la\nabla^2\wh H_N(\vec 0),\wh\bx^{\otimes 2}\ra
    .
\end{equation}
We consider the $S_{\bz}^{\perp}$-valued process initialized at $\bx^{(Q)}_{\tau}=\wh\bx_{\tau}$ and obeying for $t\geq \tau$ the dynamics:
\begin{equation}
\label{eq:xQt}
    \de \bx^{(Q)}_t 
    =
    \frac{\beta}{4} \nabla H_N^{(Q)}(\bx^{(Q)}_t)~\de t
    +
    \frac{1}{2}\de \bW_t.
\end{equation}
Crucially $\bx^{(Q)}_t$ has an Ornstein--Uhlenbeck description which can be analyzed directly.
We will show in the next subsection that if $\bx_{\tau}$ obeys \eqref{eq:approx-crit-start} and $H_N(\bx_{\tau})/N\leq E_*^{(\eps)}(H_N)$, then $\bx^{(Q)}_t$ gains energy after time $\oC\beta^{-1}$ for a large but $\beta$-independent constant $\oC=\oC(\eps)$, while closely approximating $\wh\bx_t$. 
Notice that the additional drift term $\frac{(N-3)|\sin\theta_t|}{2 N^{1/2}}\cdot \frac{\wh\bx_t}{\|\wh\bx_t\|}\,\de t$ from \eqref{eq:pushforward-Langevin} is not present in $H_N^{(Q)}$; it does not scale with $\beta$, hence does not contribute significantly to the approximation error as shown in Lemma~\ref{lem:quadratic-SDE-approximation}.

Before this, we establish several preliminary estimates.
First, it is easy to see directly that 
\begin{align}
\label{eq:same-gradient}
    \|\nabla \wh H_N(\vzero)\|&=\|\nabla_{\sph} H_N(\bx_{\tau})\|,
    \\
\label{eq:same-spectrum}
    \nabla^2 \wh H_N(\vzero)
    &\simeq 
    \nabla^2_{\sph} H_N(\bx_{\tau})
\end{align}
where $\simeq$ indicates the two matrices have the same spectrum.
Next, the chain rule easily gives the following since $H_N$ is $C$-bounded (recall that the actual value of $C$ may vary from line to line).

\begin{proposition}
\label{prop:C'-bounded}
    Both $\wh H_N$ and $H_N^{(Q)}$ are $C$-bounded on $\|\wh\bx\|\leq \sqrt{N}$, i.e. for all such $\wh\bx$:
    \begin{align*}
    \|\nabla \wh H_N(\wh\bx)\|
    &\leq 
    C\sqrt{N},
    \\
    \|\nabla^2 \wh H_N(\wh\bx)\|_{\op}
    &\leq 
    C,
    \\
    \|\nabla^3 \wh H_N(\wh\bx)\|_{\op}
    &\leq 
    C/\sqrt{N},
    \end{align*}
    and similarly for $H_N^{(Q)}$.
\end{proposition}

The next proposition is immediate by Taylor's theorem.

\begin{proposition}
\label{prop:quadratic-taylor-error}
    For $\wh\bx\in S_{\bz}^{\perp}$ with $\|\wh\bx\|\leq \sqrt{N}$, we have 
    \[
    |\wh H_N(\wh\bx)-H_N^{(Q)}(\wh\bx)|\leq CN^{-1/2}\|\wh\bx\|^3.
    \]
\end{proposition}

We next control the movement for $t\in [\tau,\tau+\oC\beta^{-1}]$. Lemma~\ref{lem:apriori-quadratic-SDE-approximation} bounds the movement from time $\tau$ using the small gradient condition \eqref{eq:approx-crit-start}, which improves the bounds in Lemma~\ref{lem:quadratic-SDE-approximation} by crucial polynomial factors of $\beta$.

\begin{lemma}
\label{lem:apriori-quadratic-SDE-approximation}
    For any constant $\oC$, suppose $\bx_{\tau}$ satisfies \eqref{eq:approx-crit-start}. 
    Then with probability $1-e^{-cN}$ \revedit{conditionally on $\cF_{\tau}$}, 
    \begin{align}
    \label{eq:mod-low-grad-low-move}
    \sup_{\tau\leq s\leq \tau+\oC\beta^{-1}}
    \|\wh\bx_{s}\|
    &\leq
    O_{\oC}(\beta^{-1/2}\sqrt{N}),
    \\
    \label{eq:Q-low-grad-low-move}
    \sup_{\tau\leq s\leq \tau+\oC\beta^{-1}}
    \|\bx^{(Q)}_{s}\|
    &\leq
    O_{\oC}(\beta^{-1/2}\sqrt{N})
    .
    \end{align}
\end{lemma}

\begin{proof}
    \revedit{
    Recall $\wh\bx_{\tau}=\bx^{(Q)}_{\tau}=\vzero\in S_{\bz}^{\perp}$.
    Let $\tau'$ be the first time that $\|\wh\bx_{\tau'}\|=\sqrt{N}$, i.e. $\theta_{\tau'}= \pi/2$. Also let $w_t=\|\wh\bx_{t\wedge \tau'}-\wh\bx_{\tau}\|^2/N$.
    The second part of Lemma~\ref{lem:movement-bound} (applied to $S_{\bz}^{\perp}\simeq \bbR^{N-1}$ rather than $\bbR^N$) shows \eqref{eq:bt-bound-near-approx-crit} with $\gamma=C\beta^{-1/2}$ using \eqref{eq:approx-crit-start}, \eqref{eq:same-gradient}. Proposition~\ref{prop:self-bounding} then yields \eqref{eq:mod-low-grad-low-move} with probability $1-e^{-cN}$, except with $\tau+\oC\beta^{-1}$ replaced by $(\tau+\oC\beta^{-1})\wedge\tau'$. By definition of $\tau'$, this implies $\tau'\geq \tau+\oC\beta^{-1}$, hence the full conclusion of \eqref{eq:mod-low-grad-low-move} (with probability $1-e^{-cN}$).
    A similar argument gives \eqref{eq:Q-low-grad-low-move}.
    }
\end{proof}

\begin{lemma}
\label{lem:quadratic-SDE-approximation}
    For any constant $\oC$, suppose $\bx_{\tau}$ satisfies \eqref{eq:approx-crit-start}. 
    Then with probability $1-e^{-cN}$ \revedit{conditionally on $\cF_{\tau}$}, 
    \begin{align}
    \label{eq:quadratic-approx-location}
    \|\bx^{(Q)}_{\tau+\oC\beta^{-1}}
    -
    \wh\bx_{\tau+\oC\beta^{-1}}\|
    \leq
    O_{\oC}(\beta^{-1}\sqrt{N}),
    \\
    \label{eq:quadratic-approx-energy}
    |\wh H_N(\bx^{(Q)}_{\tau+\oC\beta^{-1}})
    -
    \wh H_N(\wh\bx_{\tau+\oC\beta^{-1}})|
    \leq
    O_{\oC}(\beta^{-3/2}N).
    \end{align}
\end{lemma}

\begin{proof}
    First it suffices to prove \eqref{eq:quadratic-approx-location}, which together with \eqref{eq:approx-crit-start} and \eqref{eq:mod-low-grad-low-move} implies \eqref{eq:quadratic-approx-energy}.
    Indeed, these results show that with probability $1-e^{-cN}$, the line segment between $\bx^{(Q)}_{\tau+\oC\beta^{-1}}$ and $\wh\bx_{\tau+\oC\beta^{-1}}$ in $S_{\bz}^{\perp}$ has length $O_{\oC}(\beta^{-1}\sqrt{N})$, while $\|\nabla \wh H_N(\wh\bx)\|\leq O_{\oC}(\beta^{-1/2}\sqrt{N})$ holds for every $\wh\bx$ on this line segment. 
    Hence we focus on proving \eqref{eq:quadratic-approx-location}.

    Choose $\wt C$ a large constant depending on $(\eps,\oC)$ but still independent of $\beta$ (i.e. $\beta$ is much larger).
    Let $\tau'$ be the first time at which 
    $\|\wh\bx_{\tau'}\|>\wt C\beta^{-1/2}\sqrt{N}$ 
    or 
    $\|\bx^{(Q)}_{\tau'}\|>\wt C\beta^{-1/2}\sqrt{N}$.
    By Lemma~\ref{lem:apriori-quadratic-SDE-approximation}, we may restrict to the exponentially likely case $\tau'\geq \oC\beta^{-1}$, which implies that for all $\tau\leq t\leq \tau+\oC\beta^{-1}$,
    \begin{align}
    \label{eq:theta-t-small}
    \theta_t&\revedit{\in \big[0, O_{\oC}(\beta^{-1/2})\big]}
    \\
    \implies 
    \label{eq:cos-theta-t}
    1-\cos(\theta_t)&\leq O_{\oC}(\beta^{-1}).
    \end{align}
    \revedit{Let $v_t=\|\bx^{(Q)}_{t}-\wh\bx_{t}\|^2/N$.
    Using \eqref{eq:HNQ}, $t\leq \oC\beta^{-1}$, and Proposition~\ref{prop:bt-sigma-bound}, we find that $v\in \cI(b,\sigma)$ for
    $\sigma_t\leq 10 N^{-1/2}$ and}
    \[
    b_t
    \leq 
    2N^{-1}
    \lt\|\beta 
    \nabla H_N^{(Q)}(\bx^{(Q)}_t)
    -
    \beta\nabla \wh H_N(\wh\bx_t)
    +
    \frac{(N-3)|\sin\theta_t|}{2 N^{1/2}}
    \cdot \frac{\wh\bx_t}{\|\wh\bx_t\|}
    \rt\|
    \cdot
    \|
    \bx^{(Q)}_{t}-\wh\bx_{t}
    \|
    + 
    O_{\oC}(\beta^{-1}).
    \]
    Here the latter term accounts for the terms stemming from \eqref{eq:cos-theta-t}, including both the differences in drift and diffusivity.
    Similarly using \eqref{eq:theta-t-small} and Lemma~\ref{lem:apriori-quadratic-SDE-approximation},
    \[
    N^{-1}
    \lt\|
    \frac{(N-3)|\sin\theta_t|}{2 N^{1/2}}
    \cdot \frac{\wh\bx_t}{\|\wh\bx_t\|}
    \rt\|
    \cdot
    \|
    \bx^{(Q)}_{t}-\wh\bx_{t}
    \|
    \leq O_{\oC}(\beta^{-1}).
    \]
    We estimate the remaining contribution via:
    \[
    \begin{WithArrows}
    &
    \lt\|\beta 
    \nabla H_N^{(Q)}(\bx^{(Q)}_t)
    -
    \beta\nabla \wh H_N(\wh\bx_t)
    \rt\|
    \\
    &=
    \beta
    \lt\|
    \nabla\wh H_N(\vzero)
    +
    \nabla^2 \wh H_N(\vzero)\bx^{(Q)}_t
    -
    \nabla \wh H_N(\wh\bx_t)
    \rt\|
    \\
    &\leq 
    \beta 
    \lt\|
    \nabla\wh H_N(\vzero)
    +
    \nabla^2 \wh H_N(\vzero)\bx^{(Q)}_t
    -
    \nabla \wh H_N(\bx^{(Q)}_t)
    \rt\|
    +
    \beta
    \lt\|
    \nabla \wh H_N(\bx^{(Q)}_t)
    -
    \nabla \wh H_N(\wh\bx_t)
    \rt\|
    \Arrow{Prop.~\ref{prop:C'-bounded},~\ref{prop:quadratic-taylor-error}}
    \\
    &\leq 
    C\beta 
    N^{-1/2}
    \|
    \bx^{(Q)}_t
    \|^2
    +
    C\beta
    \|
    \bx^{(Q)}_t
    -
    \wh\bx_t
    \|
    \Arrow{Lem.~\ref{lem:apriori-quadratic-SDE-approximation}}
    \\
    &\leq 
    O_{\oC}(N^{1/2})
    +
    C\beta
    \|
    \bx^{(Q)}_t
    -
    \wh\bx_t
    \|
    .
    \end{WithArrows}
    \]
    All in all, we have obtained
    \[
    \de v_t
    =
    O_{\oC}\big(\beta v_t+ \sqrt{v_t}+\beta^{-1}\big) ~\de t+ O_{\oC}(N^{-1/2})\de B_t.
    \]
    Applying Proposition~\ref{prop:self-bounding} yields $\sup_{\tau\leq s\leq\tau+\oC\beta^{-1}} v_t\leq O_{\oC}(\beta^{-2})$ as desired.
\end{proof}

\begin{remark}
\label{rmk:conformal}
    Because $\cS_N$ admits a conformal local parametrization (i.e. is \emph{conformally flat}), the approximations above can be slightly simplified by introducing the process
    \[
    \de \bx^{\mod}_t
    =
    \lt(
    \beta\nabla \wh H_N(\wh\bx_t)
    -
    \frac{(N-3)|\sin\theta_t|}{2N^{1/2}}
    \cdot \frac{\bx^{\mod}_t}{\|\bx^{\mod}_t\|}
    \rt)
    \de t
    +
    \de \bW_t.
    \]
    This is a time-change of $\wh \bx_t$, and using it to approximate $\bx^{(Q)}_t$ only requires analyzing errors in the drift terms. We opted not to use this approach since it is so specialized to the sphere. 
    For example as mentioned at the end of Subsection~\ref{subsec:results}, the proofs in this section extend straightforwardly to multi-species spherical spin glasses defined on product spaces $\prod_{i=1}^r \cS_{N_r}$, which are not conformally flat in general.
    \revedit{Namely one maps to flat space using a product of stereographic projections. This retains the form of the SDE \eqref{eq:pushforward-Langevin} on each of $r$ coordinate blocks, yielding similar approximation bounds with a quadaratic potential approximation as above.}

    \revedit{In fact we believe that these arguments should extend to more general manifolds obeying suitable curvature estimates matching the order of the sphere, by similar quadratic Taylor expansion arguments. Here the curvature terms would enter similarly as small errors after fixing a good local parametrization.}
\end{remark}

\subsection{Energy Gain at Quadratic Saddle Points}
\label{subsec:escape-saddle}

We now establish an energy gain for the Ornstein--Uhlenbeck process $\bx^{(Q)}_t$. 
Interestingly the condition \eqref{eq:approx-crit-start} is not directly needed in the lemma below. It is only used to show that $\bx^{(Q)}_t$ approximates the original dynamics, and (in the outer argument for Theorem~\ref{thm:mainLB-general}) to guarantee that $\lambda_{\lfloor \eps N\rfloor}(\nabla_{\sph}^2 H_N(\bx_{\tau}))\geq \eps$ holds whenever $H_N(\bx_{\tau})$ is small. 

\begin{lemma}
\label{lem:local-improve-quadratic}
    Let $H_N$ be $C$-bounded and let $\bx_{\tau}$ satisfy $\lambda_{\lfloor \eps N\rfloor}(\nabla_{\sph}^2 H_N(\bx_{\tau}))\geq \eps$.
    Then for $\oC$ large depending on $C,\eps$ and $s=\tau+\oC\beta^{-1}$:
    \[
    \bbP\lt[
    H_N^{(Q)}(\bx^{(Q)}_{s})\geq H_N^{(Q)}(\bx^{(Q)}_{\tau})
    +
    \beta^{-1}N
    ~\big|~
    (H_N,\bx_{\tau})
    \rt]
    \geq 
    1-e^{-cN}.
    \]
\end{lemma}

\begin{proof}
    Let $\lambda_1\geq \lambda_2\geq\dots\geq \lambda_{N-1}$ be the eigenvalues for $\nabla^2 H_N^{(Q)}(\vzero)$, and $\bv_i$ the unit eigenvectors. 
    By \eqref{eq:same-spectrum} and the definition of $E_*^{(\eps)}(H_N)$, we have $\lambda_{\lfloor \eps N\rfloor}\geq \eps$. 
    Moreover $\max\big(|\lambda_1|,|\lambda_N|\big)\leq C$ since $H_N$ is $C$-bounded.
    We write
    \[
    \nabla H_N^{(Q)}(\vzero)=\sum_{i=1}^{N-1} \theta_i \bv_i,
    \quad
    \bx^{(Q)}_t = \sum_{i=1}^{N-1} a_{i,t} \bv_i.
    \]
    Below we assume all eigenvalues are non-zero; zero eigenvalues are handled identically (indeed we do \textbf{not} assume eigenvalues are bounded away from zero).
    This assumption lets us simplify the calculations by completing the square for $H_N^{(Q)}$, which amounts to translating the domain $S_{\bz}^{\perp}$ so that $\nabla H_N^{(Q)}(\vzero)=\vzero$ holds.
    Thus we may assume $\theta_i=0$, but we \emph{no longer} assume $\bx^{(Q)}_{\tau}=\vzero$. Hence for the rest of the proof, we will consider arbitrary initial conditions $(a_{1,\tau},\dots,a_{N-1,\tau})\in\bbR^{N-1}$. 

    Recalling \eqref{eq:xQt}, the evolution of $a_{i,t}$ is 
    \[
    \de a_{i,t}=\frac{\beta \lambda_i a_{i,t}}{4}\, \de t + \frac{1}{2}\de B_t.
    \]
    Given $(a_{i,\tau},\lambda_i)$ and for $Z_i$ a standard Gaussian, the conditional distribution of $a_{i,s}$ is easily seen to be that of
    \begin{equation}
    \label{eq:OU-conditional-law}
    e^{\beta\lambda_i (s-\tau)/4}a_{i,\tau}
    +
    \sqrt{
    \frac{e^{(s-\tau)\beta\lambda_i/2}-1}{2\beta\lambda_i}
    }\, Z_i
    =
    e^{\oC \lambda/4}a_{i,\tau}
    +
    \sqrt{
    \frac{e^{\oC\lambda_i/2}-1}{2\beta\lambda_i}
    }\, Z_i.
    \end{equation}
    \revedit{In particular for each $t$, the values $a_{1,t},\dots,a_{N-1,t}$ are conditionally independent Gaussians given $(H_N^{(Q)},\bx^{(Q)}_{\tau})$. 
    Hence
    \[
    H_N^{(Q)}(\bx^{(Q)}_t)
    =
    \sum_{i=1}^{N-1} 
    \lt(\frac{\lambda_i a_{i,t}^2}{2} + \theta_i a_{i,t}\rt).
    \]
    is an independent sum indexed by the $N-1$ eigenvalues which we will analyze term by term. 
    }

    First suppose that $-C\leq \lambda_i<0$.
    Since $\frac{1-e^{-x}}{x}\leq 1$ for all $x>0$, taking $x=-\oC\lambda_i/2$ we find from \eqref{eq:OU-conditional-law} that
    \begin{equation}
    \label{eq:negative-well-energy}
    \begin{aligned}
    \bbE[a_{i,s}^2~|~a_{i,\tau},\lambda_i]
    &=
    e^{\oC\lambda_i/2}a_{i,\tau}^2
    + 
    \frac{e^{\oC\lambda_i/2}-1}{2\beta\lambda_i}
    \\
    &\leq 
    a_{i,\tau}^2 
    +
    \oC\beta^{-1}
    .
    \end{aligned}
    \end{equation}
    Hence we find that 
    \begin{equation}
    \label{eq:negative-eigenvalue-bound}
    \bbE[\lambda_i a_{i,s}^2-\lambda_i a_{i,\tau}^2 ~|~
    a_{i,\tau},\lambda_i]
    \geq 
    -\lambda_i \oC/\beta\geq -C\oC/\beta.
    \end{equation}

    Next suppose $\lambda>0$.
    We again use \eqref{eq:OU-conditional-law}. Since $s-\tau=\oC\beta^{-1}$, we find
    \begin{equation}
    \label{eq:negative-eigenvalue-bound}
    \begin{aligned}
    \bbE[\lambda_i a_{i,s}^2-\lambda_i a_{i,\tau}^2 ~|~
    a_{i,\tau},\lambda_i]
    &\geq 
    \frac{e^{\oC\lambda_i/2}-1}{2\beta}
    \\
    &\geq 
    \frac{e^{\oC\eps/2}-1}{2\beta}\cdot \ind_{\lambda_i\geq \eps}
    .
    \end{aligned}
    \end{equation}
    Combining, since $\lambda_{\lfloor \eps N\rfloor}\geq \eps$ we can sum over $i$ to obtain
    \begin{align}
    \notag
    N^{-1}\cdot \bbE
    \lt[
    H_N^{(Q)}(\bx^{(Q)}_s)-H_N^{(Q)}(\bx^{(Q)}_{\tau})
    ~|~H_N^{(Q)},\bx^{(Q)}_{\tau}
    \rt]
    &=
    \frac{1}{N}
    \sum_{i=1}^{N-1}
    \bbE[\lambda_i a_{i,s}^2-\lambda_i a_{i,\tau}^2 ~|~
    a_{i,\tau},\lambda_i]
    \\
    \label{eq:crux-energy-gain}
    &\geq 
    \frac{\eps^2 (e^{\oC\eps/2}-1)}{2\beta\eps} -\frac{C\oC}{\beta}
    \geq 2C/\beta
    \end{align}
    where the last bound holds for $\oC=\oC(\eps,C)$ large.

    Finally we show exponential concentration of $H_N^{(Q)}(\bx^{(Q)}_s)$ conditionally on $(H_N^{(Q)},\bx^{(Q)}_{\tau})$.
    Let $m_i=\bbE^t[a_{i,s}]$, and as above let $Z,Z_1,\dots,Z_{N-1}$ be \iid standard Gaussians.
    For constants $A_i\leq O_{\oC}(1)$, we have
    \begin{align*}
    \sum_{i=1}^{N-1}
    \lambda_i a_{i,s}^2
    &\stackrel{d}{=}
    \sum_{i=1}^{N-1}
    \lambda_i (m_i+A_iZ_i)^2
    \\
    &=
    \sum_{i=1}^{N-1}
    \lt(\lambda_i m_i^2 
    +
    2\lambda_i m_iA_iZ_i 
    +
    \lambda_i A_i^2Z_i^2\rt).
    \end{align*}
    The contribution of $\sum_{i=1}^{N-1}
    \lambda_i m_i^2$ is non-random so it does not affect concentration. The second term contributes
    \begin{align*}
    2\sum_{i=1}^{N-1}
    \lambda_i m_iA_iZ_i
    &\stackrel{d}{=}
    2\sqrt{\sum_{i=1}^{N-1} \lambda_i^2 m_i^2 A_i^2}\, Z.
    \end{align*}
    Next, using $C$-boundedness in the last step, 
    \begin{align*}
    \sum_{i=1}^{N-1}\lambda_i^2 m_i^2 A_i^2
    &\leq 
    O_{\oC}\lt(\sum_{i=1}^{N-1} \lambda_i^2 m_i^2\rt)
    \leq 
    O_{\oC}\lt(\sum_{i=1}^{N-1} \lambda_i^2 a_{i,\tau}^2\rt)
    \\
    &\leq 
    O_{\oC}\lt(\|\nabla H_N^{(Q)}(\bx^{(Q)}_{\tau})\|^2\rt)
    \leq  O_{\oC}(N).
    \end{align*}
    Hence $2\sum_{i=1}^{N-1}\lambda_i m_i A_i Z_i$
     is a centered Gaussian with variance $O_{\oC}(N)$, so it does not contribute to the expectation and is at most $\frac{CN}{2\beta}$ in absolute value with probability $1-e^{-cN}$.

    Finally subtracting another deterministic quantity, $\sum_{i=1}^{N-1}\lambda_i A_i^2 (Z_i^2-1)$ is a sum of $N-1$ mean zero and $O_{\oC}(1)$-subexponential random variables, hence is at most $\frac{CN}{2\beta}$ in absolute value with probability $1-e^{-cN}$ by Bernstein's inequality.
    Combining with \eqref{eq:crux-energy-gain} concludes the proof.
\end{proof}

\begin{proof}[Proof of Theorem~\ref{thm:mainLB-general}]
    Given any $C$ and $\eps$, choose $\oC$ large depending on both and then $\beta$ large depending on $(C,\eps,\oC)$.
    Let $M=e^{cN/2}$ for suitably small $c$ to be chosen later. We define a sequence $0=\tau_0<\tau_1<\dots<\tau_M$ of stopping times such that with probability $1-e^{-cN}$, for each $0\leq m\leq M-1$ the following properties hold:
    \begin{align}
    \label{eq:tau-increase}
    \tau_{m+1}-\tau_m
    &\in 
    \{\beta^{-3},\oC\beta^{-1}\},
    \\
    \label{eq:energy-increase}
    H_N(\tau_{m+1})/N
    &\geq 
    \min\lt(
    E_*^{(\eps)}(H_N)-\frac{\eps}{2},
    \frac{H_N(\tau_{m})}{N}+\frac{\tau_{m+1}-\tau_m}{\revedit{2}\oC}\rt),
    \\
    \label{eq:energy-stable}
    \sup_{t,s\in [\tau_m,\tau_{m+1}]}
    |H_N(\bx_t)-H_N(\bx_s)|
    &\leq 
    C\beta^{-1/2}N.
    \end{align}
    The existence of such stopping times implies \revedit{Theorem~\ref{thm:mainLB-general}} (after adjusting $c$) since $H_N$ is uniformly bounded on $\cS_N$ (by setting $k=0$ in \eqref{eq:gradient-bounded}).
    These stopping times are constructed as follows:
    \begin{enumerate}
        \item If $H_N(\bx_{\tau_m})/N\geq E_*^{(\eps)}(H_N)$ or $\|\nabla_{\sph} H_N(\bx)\|\geq C\beta^{-1/2}\sqrt{N}$, let $\tau_{m+1}=\tau_m+\beta^{-3}$.
        \item Otherwise, let $\tau_{m+1}=\tau_m+\oC/\beta$.
    \end{enumerate}
    We analyze this construction in three cases; note that Property \eqref{eq:tau-increase} is always clear. Since $M=e^{cN/2}\ll e^{cN}$, there is no problem in union bounding over $m$.

    \subsubsection*{Case $1$: $H_N(\bx_{\tau_m})/N\geq E_*^{(\eps)}(H_N)$.}
    Applying Lemma~\ref{lem:movement-bound} with $\gamma=C$ and $s=\beta^{-3}$ shows that with probability $1-e^{-cN}$,
    \[
    \sup_{t,s\in [\tau_m,\tau_{m+1}]}
    \|\bx_t-\bx_s\|
    \leq 
    C\beta^{-3/2}.
    \]
    This implies \eqref{eq:energy-stable}, yielding the first lower bound $E_*^{(\eps)}(H_N)-\frac{\eps}{2}$ in \eqref{eq:energy-increase} for suitably large $\beta$.

    \subsubsection*{Case $2$: $\|\nabla_{\sph} H_N(\bx)\|\geq C\beta^{-1/2}\sqrt{N}$.}

    The latter two properties are immediate from Proposition~\ref{prop:gradient-improve}, where now we use the second lower bound in \eqref{eq:energy-increase}.

    \subsubsection*{Case $3$: $H_N(\bx_{\tau_m})/N\leq E_*^{(\eps)}(H_N)$ and $\|\nabla_{\sph} H_N(\bx)\|\leq C\beta^{-1/2}\sqrt{N}$.}

    We apply stereographic projection with $\bz=\bx_{\tau_m}$ to construct the processes $\wh\bx_t,\bx^{(Q)}_t$ for $t\geq \tau_m$. By definition,
    \begin{align*}
    S_{\bz}(\bx_{\tau_{m+1}})&=\wh\bx_{\tau_m+\oC\beta^{-1}}
    \\
    \implies 
    H_N(\bx_{\tau_{m+1}})
    &=
    \wh H_N
    \big(
    \wh\bx_{\tau_m+\oC\beta^{-1}}
    \big)
    .
    \end{align*}
    Lemma~\ref{lem:apriori-quadratic-SDE-approximation} implies property~\eqref{eq:energy-stable}. 
    Finally property~\eqref{eq:energy-increase} follows by writing
    \[
    \begin{WithArrows}
    H_N(\bx_{\tau_{m+1}})
    &=
    \wh H_N
    \big(
    \wh\bx_{\tau_m+\oC\beta^{-1}}
    \big)
    \Arrow{\eqref{eq:quadratic-approx-energy}}
    \\
    &\geq
    \wh H_N
    \big(
    \bx^{(Q)}_{\tau_m+\oC\beta^{-1}}
    \big)
    -
    O_{\oC}(\beta^{-3/2}N)
    \Arrow{Prop.~\ref{prop:quadratic-taylor-error}, ~\eqref{eq:Q-low-grad-low-move}}
    \\
    &\geq 
    H_N^{(Q)}
    \big(
    \bx^{(Q)}_{\tau_m+\oC\beta^{-1}}
    \big)
    -
    O_{\oC}(\beta^{-3/2}N)
    \Arrow{Lem.~\ref{lem:local-improve-quadratic}}
    \\
    &
    \geq 
    H_N^{(Q)}
    \big(
    \bx^{(Q)}_{\tau_m}
    \big)
    +
    \revedit{\frac{N}{2\beta}}
    \\
    &=
    H_N
    \big(
    \bx_{\tau_m}
    \big)
    +
    \revedit{\frac{N}{2\beta}}
    \end{WithArrows}
    \]
    where each step holds with probability $1-e^{-cN}$. 
\end{proof}

\section{Upper Bound: Proof of Theorem~\ref{thm:mainUB}}
\label{sec:UB}

In Subsection~\ref{subsec:soft-approx} below, we approximate $\bx_t$ by the soft spherical Langevin dynamics. 
In Subsection~\ref{subsec:lipschitz-approx} we use this to find a Lipschitz approximation for $\bx_t$, which implies \eqref{eq:original-UB-intro} via Proposition~\ref{prop:BOGP}.
The refinement to Theorem~\ref{thm:mainUB} is finally obtained in Subsection~\ref{subsec:complete-UB}.

\subsection{Approximation by Soft Spherical Langevin}
\label{subsec:soft-approx}

To approximate $\bx_t$ by a Lipschitz function of $H_N$, the first step is to replace the hard spherical constraint by a soft one so that the Brownian motion acts identically on all trajectories.
Let 
\begin{equation}
\label{eq:def-f}
    f_L(r)=L(r-1)^2+(r^2-1)^{p}.
\end{equation}
For $L$ large depending on $\beta$, we consider the full-space dynamics $\by_t^{(L)}$ on $\bbR^N$ defined by \revedit{$\by_t^{(L)}=\bx_0$ and}
\begin{equation}
\label{eq:lipschitz-dynamics}
    \de\by_t^{(L)} 
    = 
    \Big(\beta\nabla H_N(\by_t^{(L)}) 
    -
    f_L'(N^{-1}\|\by_t^{(L)}\|^2)\by_t^{(L)}
    \Big)
    \de t
    +
    \de \bB_t
\end{equation}
This equation has a \revedit{unique} strong solution for all positive times almost surely by \cite[Proposition 2.1]{ben2006cugliandolo}.
In Lemma~\ref{lem:lipschitz-approx} we show that for large $L$, the process $\by_t^{(L)}$ uniformly approximates $\bx_t$ on $t\in [0,T]$ with exponentially good probability.
For later use, we allow a slight generalization where $\beta$ in \eqref{eq:langevin-dynamics} and \eqref{eq:lipschitz-dynamics} is replaced by a uniformly bounded, deterministic function $\beta_t:[0,T]\to [0,\beta]$.
The proof is fully self-contained, although for fixed $\beta$ the first half establishing Equation~\eqref{eq:spherical-approx-1} also follows by \cite[Proposition 1.1]{dembo2007limiting}. 
{
We note that as in the previous section, Lemma~\ref{lem:lipschitz-approx} allows $H_N$ to be an arbitrary $C$-bounded function.
}

\begin{lemma}
\label{lem:lipschitz-approx}
    Fix any $\bx_0\in\cS_N$ and $C$-bounded $H_N:\bbR^N\to\bbR$.
    For any $p,\eps,\beta,T$ there exists $L,c$ such that the following holds. For $N$ sufficiently large, with $\bx_t,\by_t^{(L)}$ driven by the same Brownian motion and $\bx_0=\by_0^{(L)}\in \cS_N$,
    \begin{equation}
    \label{eq:full-space-approx}
    \bbP\lt[
    \sup_{t\in [0,T]} \|\bx_t-\by_t^{(L)}\|\leq \eps\sqrt{N}
    \rt]
    \geq
    1-e^{-cN}.
    \end{equation}
    Moreover \eqref{eq:full-space-approx} holds uniformly for deterministic piece-wise Lipschitz functions $\beta_t:[0,T]\to [0,\beta]$.
\end{lemma}

\begin{proof}
    Choose $\delta=\delta(p,\eps,\beta,T)$ small and $L=L(p,\eps,\beta,T,\delta)$ large. Define the (non-Markov) process 
    \begin{equation}
    \label{eq:wt-by}
    \wt \by_t^{(L)} = \frac{N^{1/2} \by_t^{(L)}}{\|\by_t^{(L)}\|}
    \end{equation}
    which is normalized to be $\cS_N$-valued.
    We will prove that with probability $1-e^{-cN}$ for some $c>0$:
    \begin{align}
    \label{eq:spherical-approx-1}
    \sup_{t\in [0,T]} \|\wt\by_t^{(L)}-\by_t^{(L)}\|&\leq \sqrt{\delta N},
    \\
    \label{eq:spherical-approx-2}
    \sup_{t\in [0,T]} \|\bx_t-\wt\by_t^{(L)}\|&\leq \eps\sqrt{N}/2
    .
    \end{align}
    Combining via the triangle inequality will then imply the result.

    We first show \eqref{eq:spherical-approx-1} holds with sufficiently high probability. Letting 
    \[
    z_t=\lt(N^{-1/2}\|\by_t^{(L)}\|-1\rt)^2,
    \]
    \eqref{eq:spherical-approx-1} is equivalent to the estimate $\sup_{t\in [0,T]}z_t\leq \delta$.
    Thus let $\tau$ be the first time that $z_t\geq \delta$, so Proposition~\ref{prop:gradients-bounded-0}\ref{it:spherical-grad-bounded} applies to $\by_t^{(L)}$ before time $\tau$. For $t\leq \tau$, \revedit{Proposition~\ref{prop:bt-sigma-bound}} gives \revedit{$z\in\cI(b,\sigma)$} with $\sigma_t=2\sqrt{z_t/N}$ and  
    \[
    b_t\leq -2Lz_t+O(\beta+1)\sqrt{z_t}.
    \]
    Here the first term was obtained from 
    \[
    |f_L'(1+\sqrt{z_t})|,|f_L'(1-\sqrt{z_t})|\geq 2L\sqrt{z_t}
    \]
    The $O(\beta+1)$ encapsulates $\nabla H_N(\by^{(L)}_t)$ and the \Ito\ term. For $t\leq \tau$ and $L$ large depending on $\beta$, AM-GM gives
    \[
    b_t
    \leq  -2Lz_t + O(\beta+1)\sqrt{z_t}
    \leq -Lz_t + 1.
    \]

    Next with $\eta=\delta/10$, let $M_t = e^{ \eta Lz_tN}$. Then $M_{t\wedge \tau}$ is uniformly bounded for any fixed $N$, and we \textbf{claim} that 
    \[
    A_t\equiv M_{t\wedge\tau}+ (T-t)e^{2\eta^2 LN}
    \]
    is a super-martingale. This follows from bounding the drift term of $\de M_t$ by
    \begin{align*}
    \lt( \eta LNb_t + \frac{(\eta L N \sigma_t)^2}{2}\rt)M_t
    &\leq 
    \lt(
    \eta L(-Lz_t+1)
    +
    2\eta^2 L^2 z_t
    \rt) 
    NM_t 
    \\
    &\stackrel{(\dagger)}{\leq}
    \lt(\eta L^{-1} + 2\eta^3 \rt) L^2 N e^{\eta^2  LN}
    \\
    &\leq 
    \eta L^2 N e^{\eta^2  LN}
    \\
    &\leq e^{2\eta^2 LN}.
    \end{align*}
    Here the step $(\dagger)$ follows because the former expression is negative whenever $z_t\geq \eta$ and $\eta\leq 0.1$ (since $L$ is large depending on $\delta$). The last step holds for $N$ sufficiently large.

    In addition to being a super-martingale, $A_t$ is almost surely non-negative and
    \[
    A_0=Te^{2\eta^2 LN}+1\leq 2T e^{2\eta^2 LN}. 
    \]
    Further if $\tau\leq T$ then
    \[
    A_{\tau}
    \geq 
    M_{\tau}=e^{ \eta L z_{\tau}N}= e^{\eta L\delta N}.
    \]
    Since $\eta=\delta/10$, optional stopping implies
    \[
    \bbP[\tau\leq T]\leq A_0 e^{-\eta L\delta N} 
    \leq 
    2T e^{(2\eta^2 - \eta \delta) LN}
    \leq 
    2Te^{-\delta^2 LN/20}.
    \]
    This implies \eqref{eq:spherical-approx-1}, in fact even on exponential time-scales.

    Next we turn to \eqref{eq:spherical-approx-2}. We will show it holds with probability $1-e^{-cN}$ with $t\in [0,T]$ replaced by $t\in [0,T\wedge\tau]$, which suffices in light of \eqref{eq:spherical-approx-1}.
    Thus 
    \begin{equation}
    \label{eq:rt}
    r_t\equiv N^{-1/2} \|\bx_t\|\in \lt[1-\delta^{1/2},1+\delta^{1/2}\rt]   
    \end{equation}
    holds below.
    The point is that $\bx_t$ and $\wt \by_t^{(L)}$ obey almost the same stochastic differential equation. 
    Namely the latter satisfies the (non-autonomous) equation
    \[
    \de \wt \by_t^{(L)}
    =
    \lt(
    r_t^{-1}
    \beta\nabla_{\sph} H_N(\by_t^{(L)}) 
    -
    \frac{(N-1)\wt \by_t^{(L)}}{2N}
    \rt)\de t
    +
    r_t^{-1} P_{\wt \by_t^{(L)}}^{\perp} \de \bB_t.
    \]
    (Recall that \eqref{eq:spherical-derivative} defines $\nabla_{\sph}H_N(\cdot)$ even for $\by_t^{(L)}\notin \cS_N$).
    Note the soft spherical term $f_L'(\cdot)$ acts radially, and hence does not appear explicitly in the above equation since $\wt \by_t^{(L)}$ is constrained to lie on $\cS_N$.
    Hence we find 
    \begin{align*}
    \de (\bx_t-\wt \by_t^{(L)})
    &=
    \beta\lt(\nabla_{\sph} H_N(\bx_t)
    -
    r_t^{-1}\nabla_{\sph} H_N(\by_t^{(L)})\rt) 
    \de t
    +
    \frac{N-1}{2N}
    (\wt\by_t^{(L)}-\bx_t)\,\de t
    \\
    &\quad\quad
    +
    \lt(
    P_{\bx_t}^{\perp}-r_t^{-1}P_{\wt \by_t^{(L)}}^{\perp} 
    \rt)
    \,\de \bB_t.
    \end{align*}
    By Proposition~\ref{prop:bt-sigma-bound}, the scalar process $w_t=N^{-1}\|\bx_t-\wt\by_t^{(L)}\|^2$ \revedit{lies in $\cI(b,\sigma)$ for}:
    \begin{align*}
    b_t
    &=
    2\beta N^{-1}\,
    \lt\la
    \nabla_{\sph} H_N(\bx_t)-r_t^{-1}\nabla_{\sph} H_N(\by_t^{(L)})
    ,
    \bx_t-\wt\by_t^{(L)}
    \rt\ra
    -
    \frac{N-1}{N}w_t
    \\
    &
    \quad\quad
    +
    N^{-1}
    \lt\|
    P_{\bx_t}^{\perp}-r_t^{-1}P_{\wt \by_t^{(L)}}^{\perp} 
    \rt\|_F^2
    ,
    \\
    |\sigma_t|
    &=
    2N^{-1}\lt|
    \lt\la
    \bx_t-\wt\by_t^{(L)}
    ,
    P_{\bx_t}^{\perp}-r_t^{-1}P_{\wt \by_t^{(L)}}^{\perp} 
    \rt\ra
    \rt|
    \\
    &\leq 5N^{-1/2}
    \lt\|P_{\bx_t}^{\perp}-r_t^{-1}P_{\wt \by_t^{(L)}}^{\perp}\rt\|_{\op}
    \\
    &\leq 
    10 N^{-1/2}
    .
    \end{align*}
    We estimate $b_t$ assuming $t\leq \tau$.
    First it is easy to see by $C$-boundedness that
    \begin{equation}
    \label{eq:drift-error-rt}
    \|\nabla_{\sph} H_N(\bx_t)-r_t^{-1}\nabla_{\sph} H_N(\by_t^{(L)})\|
    \leq 
    C\sqrt{N(w_t+\delta)}
    \end{equation}
    for $r_t$ satisfying \eqref{eq:rt}. 
    The matrix 
    $P_{\bx_t}^{\perp}-P_{\wt \by_t^{(L)}}^{\perp}$
    has rank at most $2$ and operator norm $O(\sqrt{w_t})$, so it follows that 
    \[
    \lt\|
    P_{\bx_t}^{\perp}-r_t^{-1}P_{\wt \by_t^{(L)}}^{\perp} 
    \rt\|_F
    \leq 
    O(\sqrt{N(w_t+\delta)}).
    \]
    Combining, using $\|\bx_t-\wt \by_t^{(L)}\|=\sqrt{N w_t}$, and dropping the second contribution to $b_t$ since $w_t\geq 0$, we find 
    \begin{equation}
    \label{eq:b-t-bound}
    b_t
    \leq 
    C(\beta+1)(w_t+\delta).
    \end{equation}
    Finally applying Proposition~\ref{prop:self-bounding} with $C_1=C(\beta+1),C_2=C\sqrt{(\beta+1)\delta}$ and $C_3=C(\beta+1)\delta$ yields that
    \[
    \sup_{t\in [0,T\wedge \tau]}
    w_t
    \leq e^{C\beta T} O(\beta^2+1)\delta
    \leq 
    \eps/2
    \]
    holds with probability $1-e^{-cN}$.
    The latter step is by our choice of $\delta$ and completes the proof.
\end{proof}

\subsection{Lipschitz Approximation}
\label{subsec:lipschitz-approx}

Next we show that $\by_T^{(L)}$ agrees with a Lipschitz function $\ell_{\bB}^{(L)}(H_N)$ of $H_N$ with probability at least $1-e^{-cN}$.
The technique below is similar to that of \cite[Lemma 2.5]{ben2006cugliandolo}, which was called \emph{localized concentration of measure} therein.
Note that the probability in \eqref{eq:soft-langevin-lipschitz} is averaged over the randomness of both $H_N$ (as in \eqref{eq:def-hamiltonian}) and $\bB_{[0,T]}$.

\begin{lemma}
\label{lem:soft-langevin-lipschitz}
    For $L\geq L(\beta,\eps)$ and fixed $\bx_0$, there exists for \ae $\bB_{[0,T]}$ a $C(L,\eps)$-Lipschitz $\ell_{\bB_{[0,T]}}^{(L)}:\sH_N\to \cB_N$ with
    \begin{equation}
    \label{eq:soft-langevin-lipschitz}
        \bbP\lt[(1+\eps)\ell_{\bB_{[0,T]}}^{(L)}(H_N)=\by_T^{(L)}\rt]\geq 1-e^{-cN}
        .
    \end{equation}
    Moreover this holds uniformly over deterministic piece-wise Lipschitz $\beta_t:[0,T]\to [0,\beta]$.
\end{lemma}

\begin{proof}[Proof of Lemma~\ref{lem:soft-langevin-lipschitz}]
    Let $L,\eps$ be as in Lemma~\ref{lem:lipschitz-approx}. We \textbf{claim} that for almost all $\bB_{[0,T]}$, the map
    \[
    H_N\mapsto\by_T^{(L)}
    \]
    is $C(T,\beta,L)$-Lipschitz when restricted to the set $\wt K_N(\bB)\subseteq K_N$ of Hamiltonians for which $\sup_{t\in [0,T]} \|\bx_t-\by_t^{(L)}\|\leq \sqrt{N}$. 
    Assuming this claim, by the Kirzsbraun extension theorem there exists for each $\bB_{[0,T]}$ a Lipschitz function $\ell_{\bB}^{(L)}:\sH_N\mapsto \cB_N$ such that $\ell_{\bB_{[0,T]}}(H_N)=(1+\eps)^{-1}\by_T^{(L)}$ whenever $H_N\in \wt K_N(\bB)$.\footnote{The proof of Kirzsbraun's extension theorem is constructive so there is no issue of measurable dependence on $\bB$.}
    Moreover it follows from Proposition~\ref{prop:gradients-bounded} and Lemma~\ref{lem:lipschitz-approx} that 
    \[
    \bbP[H_N\in \wt K_N(\bB)]\geq 1-e^{-cN}.
    \]
    Therefore it remains to prove the above claim.
    
    Let $H_{N,1},H_{N,2}\in \wt K_N(\bB)$ with $\|H_{N,1}-H_{N,2}\|_2\leq \Delta$, and let the corresponding Langevin dynamics be $\by_{t,1}^{(L)},\by_{t,2}^{(L)}$ with the same driving Brownian motion $\bB_t$.
    With $W_t=\|\by_{t,1}^{(L)}-\by_{t,2}^{(L)}\|^2$, \Ito's formula gives
    \begin{align*}
    \frac{1}{2}\de W_t
    &=
    \lt\la 
    \by_{t,1}^{(L)}-\by_{t,2}^{(L)},\de \by_{t,1}^{(L)}-\de \by_{t,2}^{(L)}
    \rt\ra
    \,\de t
    \\
    &=
    \beta_t
    \lt\la 
    \by_{t,1}^{(L)}-\by_{t,2}^{(L)},
    \nabla H_{N,1}(\by_{t,1}^{(L)})-\nabla H_{N,2}(\by_{t,2}^{(L)})
    \rt\ra
    \\
    &\quad\quad
    -
    L
    \lt\la
    \by_{t,1}^{(L)}-\by_{t,2}^{(L)},
    f'\big(N^{-1}\|\by_{t,1}^{(L)}\|^2\big)\by_{t,1}^{(L)}-f'\big(N^{-1}\|\by_{t,2}^{(L)}\|^2\big)\by_{t,2}^{(L)}
    \rt\ra
    .
    \end{align*}
    (Note that the Brownian motion contributions to $\by_{t,1}^{(L)},\by_{t,2}^{(L)}$ cancel.)
    Next for $C$-bounded $H_{N,1},H_{N,2}\in K_N$ we have 
    \[
    \|\nabla H_{N,1}(\by_{t,1}^{(L)})
    -
    \nabla H_{N,2}(\by_{t,1}^{(L)})\|\leq O(\|\by_{t,1}^{(L)}-\by_{t,2}^{(L)}\| + \Delta).
    \]
    Since $f$ is smooth (recall \eqref{eq:def-f}), AM-GM implies the differential inequality:
    \begin{equation}
    \label{eq:diff-ineq}
    \de W_t
    \leq 
    O_{\beta,L}(W_t+\Delta^2).
    \end{equation}
    \revedit{Since $W_0=0$}, Gr{\"o}nwall implies $W_T\leq O_{T,\beta,L}(\Delta^2)$, establishing the above claim and finishing the proof.
\end{proof}

\begin{remark}
\label{rem:reflection}
    $\ell_{\bB}^{(L)}$ can be constructed explicitly for $C$-bounded $H_N$ by constraining $\by_t^{(L)}$ to lie in the ball of radius $2\sqrt{N}$ via inward orthogonal reflection. Indeed inward orthogonal reflection for a convex body only improves the differential inequality \eqref{eq:diff-ineq} (see \cite[Proof of Theorem 11]{huang2021tight}), and this reflection has no effect with $1-e^{-cN}$ probability by Lemma~\ref{lem:lipschitz-approx}. However Kirzsbraun's theorem would still be needed to define $\ell_{\bB}^{(L)}$ for $H_N\notin K_N$. 
\end{remark}

We are now ready to prove \eqref{eq:original-UB-intro} using Proposition~\ref{prop:BOGP}.
{We note that in the proof below, the fact that $H_N$ is the Hamiltonian of a $p$-spin model (rather than a deterministic $C$-bounded function) enters only via \cite{huang2023algorithmic}.
}

\begin{proof}[Proof of \eqref{eq:original-UB-intro}]
    We treat $\bx_0\in\cS_N$ as deterministic below since it is assumed to be independent of $H_N$.
    First for each fixed $t\in [0,T]$, we can apply Lemma~\ref{lem:soft-langevin-lipschitz} to obtain a $C(T,L,\beta)$-Lipschitz map $\ell_{\bB_{[0,t]}}^{(L)}:\sH_N\to\cB_N$.
    Then by
    \cite[Theorem 1 and Corollary 1.8]{huang2023algorithmic} we obtain
    \[
    \bbP[H_N(\ell_{\bB_{[0,t]}}(H_N))\leq E_{\infty}+\eta]\geq 1-e^{-cN}.
    \]
    Take $\eps$ small depending on $\eta$, and $L$ as in Lemma~\ref{lem:lipschitz-approx}.
    With probability $1-e^{-cN}$ in all steps, we find 
    \[
    \begin{WithArrows}
    H_N(\ell_{\bB}^{(L)}(H_N))/N&\leq E_{\infty}+\frac{\eta}{4}
    \Arrow{$H_N\in K_N$}
    \\
    \implies
    H_N((1+\eps)\ell_{\bB}^{(L)}(H_N))/N&\leq E_{\infty}+\frac{\eta}{2}
    \Arrow{$Lem.~\ref{lem:soft-langevin-lipschitz}$}
    \\
    \implies
    H_N(\by_t^{(L)})/N&\leq E_{\infty}+\frac{\eta}{2}
    \Arrow{$Lem.~\ref{lem:lipschitz-approx},~H_N\in K_N$}
    \\
    \implies
    H_N(\bx_t)/N&\leq E_{\infty}+\frac{3\eta}{4}.
    \end{WithArrows}
    \]

    Finally we show the bound holds simultaneously for all $t$. 
    Indeed for $\alpha$ a small constant (potentially depending on all preceding constants except $N$), a union bound shows that 
    \[
    \sup_{t\in \alpha\bbZ\cap [0,T]}
    H_N(\bx_t)/N\leq E_{\infty}+\frac{3\eta}{4}
    \]
    with probability $1-e^{-cN}$. 
    Then applying Lemma~\ref{lem:movement-bound} with $s=\alpha$ implies that
    \[
    \bbP\lt[
    \sup_{t\in [k\alpha,(k+1)\alpha]}
    H_N(\bx_t)
    \leq 
    H_N(\bx_{k\alpha})+\frac{N\eta}{4}
    \rt]
    \geq 
    1-e^{-cN}
    \]
    for all $k\in \big\{0,1,\dots, \lfloor T\alpha^{-1}\rfloor\big\}$, concluding the proof.
\end{proof}

\subsection{Completing the Proof of Theorem~\ref{thm:mainUB}}
\label{subsec:complete-UB}

\begin{proposition}
\label{prop:energy-concentrates-exponentially}
    For any $\beta,t$ and \revedit{deterministic} $\bx_0\in \cS_N$ the energy $H_N(\bx_t)/N$ concentrates exponentially in the sense that for some deterministic function $F_N$ and any $\eps>0$,
    \[
    \bbP\lt[\lt|
    \frac{H_N(\bx_t)}{N}-F_N(\beta,t)
    \rt|\leq \eps\rt]\geq 1-e^{-cN}.
    \]
\end{proposition}

\begin{proof}
    By \cite[Equation (2.21) and Proof of Lemma 2.7]{ben2006cugliandolo}, for each $L$ the energy of soft spherical Langevin dynamics concentrates exponentially, i.e. for some deterministic function $F_N$,
    \[
    \bbP\lt[\lt|
    \frac{H_N(\by^{(L)}_t)}{N}-F_N(\beta,t)
    \rt|\leq \eps/2\rt]\geq 1-e^{-cN}.
    \]
    Indeed $H_N(\by^{(L)}_t)/N$ equals the function $A_N(t,t)$ defined in Equation (1.15) therein, up to a fixed constant factor.\footnote{For extensions to mixed models, one could proceed similarly by defining degree-by-degree versions of $A_N(t,t)$ in \cite{ben2006cugliandolo}.}
    In particular, for any $\beta,\eps,t$ we may choose $L$ large as in Lemma~\ref{lem:lipschitz-approx}.
    Then with probability $1-e^{-cN}$ we also have 
    \[
    \|\bx_t-\by^{(L)}_t\|\leq \frac{\eps\sqrt{N}}{2C}.
    \]
    Since $H_N$ is $C$-bounded with exponentially good probability, combining yields the desired result.
\end{proof}

\begin{proposition}
\label{prop:gradient-LB-aukosh}
    For $\beta\geq 0$ there is $\eta=\eta(\beta,p)$ such that for any $T\geq 0$, with probability $1-o_N(1)$,
    \[
    \inf_{t\in [0,T]}
    N^{-1/2}\|\nabla_{\sph} H_N(\bx_t)\|
    \geq 
    \eta.
    \]
\end{proposition}

\begin{proof}
    We use \cite[Theorem 1.1]{ben2020bounding} which states the following when $\bx_0$ is independent of $H_N$. With $\cP(\cdot)$ denoting a space of probability measures, every subsequential limit of 
    \[
    (u_N(t),v_N(t))=\big(H_N(\bx_t)/N,\|\nabla_{\sph} H_N(\bx_t)\|^2/N\big)
    \]
    in $\cP(C([0,T])^2)$ is a.s. continuously differentiable and uniformly bounded. 
    Further it satisfies
    \begin{align}
    \label{eq:init-BGJ}
    (u(0),v(0))&=(0,p),
    \\
    \label{eq:v'-BGJ}
    v'(t)
    &\geq 
    p(p-1)+p^2 u^2 -\Theta(p,\beta,u)v
    \end{align}
    for $\Theta=p-1+2\beta pu+2\beta \Lambda_p$, with $\Lambda_p$ defined there.\footnote{Our convention differs from \cite{ben2020bounding} in the sign of $H_N$, as well as the scaling of time (our dynamics with $\beta$ corresponds to theirs for $2\beta$, but run at half speed). Hence \eqref{eq:v'-BGJ} looks slightly different from what is written there.
    The initial condition \eqref{eq:init-BGJ} holds because $\bx_0$ is independent of $H_N$: we have $H_N(\bx_0)/N\sim \cN(0,1/N)$ while conditionally on $\bx_0$, $p^{-1/2}\nabla_{\sph}H_N(\bx_0)\in S_{\bx_0}^{\perp}$ is standard Gaussian by e.g.\ \cite[Lemma 3.2 (b)]{belius2022triviality}.}
    It follows that $v'(t)>0$ when $v$ is sufficiently small. In particular for some $\eta>0$, we have $v(t)\geq 2\eta$ for all $t\in [0,T]$ in all possible subsequential limits. This implies the desired result. 
\end{proof}

We note that \cite[Theorem 1.2]{ben2020bounding} gives a more explicit lower bound than Proposition~\ref{prop:gradient-LB-aukosh}, but its statement is restricted to $t\in [T_0,T]$ rather than $t\in [0,T]$.

\begin{proof}[Proof of Theorem~\ref{thm:mainUB}]

Let $\eta$ be as in Proposition~\ref{prop:gradient-LB-aukosh}. We fix $t$ and show:
\begin{equation}
    \label{eq:improvedUB-fixed-time}
    \bbP\lt[
    H_N(\bx_t)/N
    \leq
    E_{\infty}-\delta
    \rt]
    \geq 1-e^{-cN},\quad\forall~t\in [0,T].
\end{equation}
Simultaneity over $t\in [0,T]$ then follows as in the proof of Theorem~\ref{thm:mainUB}. Namely one union bounds over $t\in \big\{0,\alpha,2\alpha,\dots,\alpha \lfloor T/\alpha\rfloor\big\}$ for a sufficiently small constant $\alpha$ and applies Lemma~\ref{lem:movement-bound} on each interval $[j\alpha,(j+1)\alpha]$.

Fixing $t$, we show \eqref{eq:improvedUB-fixed-time} by lowering the temperature at time $t$.
More precisely, we define
\[
    \wt\beta = 10C\eta^{-1}
\]
where as usual $C$ is a large constant.
Then, let $(\wt\bx_s)_{s\geq t}$ be the spherical Langevin dynamics at inverse temperature $\wt\beta$ started from $\wt\bx_t=\bx_t$. By Proposition~\ref{prop:gradient-LB-aukosh} and Equation~\eqref{eq:energy-gain-simple}, we find
\[
    \bbP\lt[
    H_N(\wt\bx_{t+\wt\beta^{-3}})-H_N(\bx_t)
    \geq 
    2\delta N
    \rt]
    \geq 1-o_N(1)
\]
for some $\delta\geq \Omega(\eta^3)>0$.
On the other hand, recall that the proof of \eqref{eq:original-UB-intro} applies to any time-inhomogenous Langevin process of the form \eqref{eq:langevin-dynamics} for any deterministic piece-wise Lipschitz $\beta_t$ lying in $[0,\beta]$.
In particular, 
\[
    \bbP[H_N(\wt\bx_{t+\wt\beta^{-3}})/N\leq E_{\infty}+\delta]\geq 1-e^{-cN}.
\]
The previous two displays together yield
\[
    \bbP\lt[
    H_N(\bx_t)/N
    \leq 
    E_{\infty}-\delta
    \rt]
    \geq 1-o_N(1).
\]
Applying 
Proposition~\ref{prop:energy-concentrates-exponentially} improves the probability to $1-e^{-cN}$ as desired.
\end{proof}

\section{Gradient Upper Bound and Cugliandolo--Kurchan Equations}

Here we prove the remaining claims, namely Corollary~\ref{cor:grad-small} and the Cugliandolo--Kurchan equations for the hard spherical dynamics (Theorem~\ref{thm:CK-new}).
We note that Corollary~\ref{cor:grad-small} relies on the matching thresholds in Theorems~\ref{thm:mainLB} and \ref{thm:mainUB} so it works only for pure models. However Theorem~\ref{thm:CK-new} generalizes with no changes to mixed $p$-spin models without external field; in the case of nonzero external field an analogous extension of \cite{zamfir2008limiting} is possible.

\subsection{Proof of Corollary~\ref{cor:grad-small}}

The next lemma is similar to Proposition~\ref{prop:gradient-improve} but at a general scale to enable a $\beta$-independent energy gain.

\begin{lemma}
\label{lem:conditional-energy-increase}
Suppose $H_N$ is $C$-bounded and the disorder-dependent $\bx_0\in \cS_N$ satisfies $\|\nabla_{\sph} H_N(\bx_0)\|\geq \eta\sqrt{N}$ and $\beta\geq 100C\eta^{-2}$. 
Then
\begin{equation}
\label{eq:conditional-energy-increase}
    \bbP\lt[
    \sup_{0\leq t\leq C/\beta}
    H_N(\bx_t)
    \geq 
    H_N(\bx_0)+N\delta
    \rt]
    \geq 
    1-e^{-cN}
\end{equation}
for sufficiently small $\delta=\delta(p,\eta)>0$ not depending on $\beta$).
\end{lemma}

\begin{proof}
    Choose $\delta$ such that $10C\delta^{1/3}\leq \eta$.
    By Lemma~\ref{lem:movement-bound} with $s=\delta/\beta$ and $\gamma=C$, with probability $1-e^{-cN}$,
    \[
    \sup_{0\leq t\leq \delta/\beta}
    \|\bx_t-\bx_0\|
    \leq 
    \delta^{1/3} \sqrt{N}.
    \]
    On this event, we have
    $\|\nabla_{\sph} H_N(\bx_t)\|\geq \eta\sqrt{N}/2$ for all $0\leq t\leq \delta/\beta$ by choice of $\delta$. 
    Recalling \eqref{eq:d-HN}, $H_N(\bx_t)$ \revedit{lies in $\cI(b,\sigma)$} for
    \[
    b_t
    \geq \lt(\frac{\eta^2\beta}{4} -C\rt)N\geq \frac{\eta^2\beta N}{8}
    \]
    and $|\sigma_t|\leq CN^{-1/2}$ (recall $\beta\geq 100C\eta^{-2}$).
    The result now follows by Proposition~\ref{prop:almost-monotone} (after adjusting $\delta$).
\end{proof}

\begin{proof}[Proof of Corollary~\ref{cor:grad-small}]
    Suppose for sake of contradiction that Corollary~\ref{cor:grad-small} is not true. Let $\tau$ be the first time in $[T_0,T]$ such that $\|\nabla_{\sph}H_N(\bx_t)\|\geq \eta\sqrt{N}$, and condition on $\cF_{\tau}$ restricted to the event that $\tau\leq T$.
    Taking $\beta,\delta$ as in Lemma~\ref{lem:conditional-energy-increase} with $\beta$ sufficiently large, we find that 
    \[
    \bbP\lt[
    \sup_{s,t\in [T_0,T+C\beta^{-1}]}
    \frac{H_N(\bx_t)-H_N(\bx_s)}{N}
    \geq \delta 
    ~\big|~
    \cF_{\tau}
    \rt]
    \geq 1/2.
    \]
    However Theorems~\ref{thm:mainLB} and \ref{thm:mainUB} together imply (since $\delta$ is independent of $\beta$) that
    \[
    \bbP\lt[
    \sup_{s,t\in [T_0,T+C\beta^{-1}]}
    \frac{H_N(\bx_t)-H_N(\bx_s)}{N}
    \geq \delta
    \rt]
    \leq e^{-cN}.
    \]
    We conclude that $\bbP[\tau\leq T]\leq 2e^{-cN}$ which completes the proof.
\end{proof}

\subsection{Deducing the Cugliandolo--Kurchan Equations}
\label{subsec:CK}

In Lemma~\ref{lem:lipschitz-approx}, we showed the hard spherical dynamics are well-approximated by a soft spherical modification of the dynamics. Such modifications were studied in \cite{ben2006cugliandolo} which proved the validity of the Cugliandolo--Kurchan equations for the correlation and integrated response functions. 
Further, the $L\to\infty$ limit of these equations was identified in \cite{dembo2007limiting}, but it was not shown that this limit actually describes the behavior of the hard spherical dynamics \eqref{eq:langevin-dynamics}.
We establish this description below by combining the aforementioned results with Lemma~\ref{lem:lipschitz-approx}.

Recall the dynamics $\by_t^{(L)}$ defined by \eqref{eq:lipschitz-dynamics}.
The Cugliandolo--Kurchan description for $\by_t^{(L)}$ relates certain two-time observables through a non-linear integro-differential equation recalled in the next two definitions.

\begin{definition}
\label{def:empirical-correlation}
    For fixed $\beta$ and $L\geq 0$, define the random functions
    \begin{align*}
    C_N(s,t)&=\la \bx_s,\bx_t\ra/N,\quad\quad\quad\quad
    \chi_N(s,t)=\la \bx_s,\bB_t\ra/N,
    \\
    C_N^{(L)}(s,t)&=\la \by_s^{(L)},\by_t^{(L)}\ra/N,
    \quad\quad
    \chi_N^{(L)}(s,t)=\la \by_s^{(L)},\bB_t\ra/N.
    \end{align*}
\end{definition}

\begin{definition}
\label{def:CK}
    Let $\bbR^2_{\geq}=\{(s,t)\in\bbR_{\geq 0}^2~:~s\geq t\}$.
    For fixed $\beta$ and $L\geq 0$, let $R^{(L)},C^{(L)}:\bbR^2_{\geq}\to\bbR$ and $K^{(L)}:\bbR_{\geq 0}\to \bbR$ satisfy $R^{(L)}(s,s)=1$, $K^{(L)}(s)=C^{(L)}(s,s)$ and solve the system:
    \begin{align*}
    \partial_s R^{(L)}(s,t)
    &=
    -f_L'(K^{(L)}(s))R^{(L)}(s,t)
    +
    \beta^2 p(p-1)
    \int_t^s R^{(L)}(u,t)R^{(L)}(s,u)C^{(L)}(s,u)^{p-2}~\de u,
    \\
    \partial_s C^{(L)}(s,t)
    &=
    -f_L'(K^{(L)}(s))C^{(L)}(s,t)
    +
    \beta^2 p(p-1)
    \int_0^s C^{(L)}(u,t)R^{(L)}(s,u)C^{(L)}(s,u)^{p-2}~\de u
    \\
    &\quad\quad
    +
    \beta^2 p\int_0^t C^{(L)}(s,u)^{p-1} R^{(L)}(t,u)~\de u,
    \\
    \partial_s K^{(L)}(s)
    &=
    -2f_L'(K^{(L)}(s))K^{(L)}(s)
    +
    1
    +
    2\beta^2 p^2
    \int_0^s C^{(L)}(s,u)^{p-1}R^{(L)}(s,u)~\de u.
    \end{align*}
    Moreover let $R,C:\bbR^2_{\geq}\to \bbR$ satisfy $R(s,s)=C(s,s)=1$ and solve the system:
    \begin{align*}
    \partial_s R(s,t)
    &=
    -\mu(s)R(s,t)
    +
    \beta^2 p(p-1)
    \int_t^s R(u,t)R(s,u)C(s,u)^{p-2}~\de u,
    \\
    \partial_s C(s,t)
    &=
    -\mu(s)C(s,t)
    +
    \beta^2 p(p-1)
    \int_0^s C(u,t)R(s,u)C(s,u)^{p-2}~\de u
    \\
    &\quad\quad
    +
    \beta^2 p\int_0^t C(s,u)^{p-1} R(t,u)~\de u;
    \\
    \mu(s)
    &\equiv 
    \frac{1}{2}
    +
    \beta^2 p^2 
    \int_0^s C(s,u)^{p-1} R(s,u)~\de u.
    \end{align*}
    Finally define $\chi^{(L)},\chi:\bbR^2_{\geq}\to\bbR$ via
    $\chi^{(L)}(s,t)
    =
    \int_0^t R^{(L)}(s,u)~\de u$
    and 
    $\chi(s,t)
    =
    \int_0^t R(s,u)~\de u
    $.
\end{definition}

We now recall the Cugliandolo--Kurchan equations for the soft dynamics $\by_t^{(L)}$, and their $L\to\infty$ limit as bivariate continuous functions.
{Note that \cite{ben2006cugliandolo,dembo2007limiting} consider the same dynamics as us but for $-H_N(\cdot)$; this negation affects none of the quantites tracked below.}

\begin{proposition}[{\cite[Theorem 1.2]{ben2006cugliandolo},~\cite[Proof of Proposition 1.1]{dembo2007limiting}}]
\label{prop:soft-convergence}
    Both systems of equations in Definition~\ref{def:CK} have a unique continuous, locally bounded solution.
    Moreover uniformly on $0\leq t\leq s\leq T$ as $L\to\infty$:
    \[
    \Big(R^{(L)}(s,t),C^{(L)}(s,t),K^{(L)}(s)\Big)
    \to 
    \Big(R(s,t),C(s,t),1\Big)
    .
    \]
    Finally for any $T,L>0$, almost surely uniformly on $s,t\in [0,T]^2$ as $N\to\infty$:
    \[
    \chi_N^{(L)}(s,t)\to \chi^{(L)}(s,t),
    \quad\quad 
    C_N^{(L)}(s,t)\to C^{(L)}(s,t).
    \]
\end{proposition}

We now obtain the Cugliandolo--Kurchan equations for the hard spherical Langevin dynamics.
Similarly to Subsection~\ref{subsec:complete-UB}, the results of \cite[Section 2]{ben2006cugliandolo} could be used to obtain exponential concentration estimates below.

\begin{theorem}
\label{thm:CK-new}
    For $T>0$, almost surely uniformly on $s,t\in [0,T]^2$ as $N\to\infty$:
    \begin{align*}
    \chi_N(s,t)&\to \chi(s,t),
    \\
    C_N(s,t)&\to C(s,t).
    \end{align*}
\end{theorem}

\begin{proof}
    We focus on the convergence of $C_N(s,t)$ as the proof for $\chi_N$ is analogous. 
    Fixing $\delta\in (0,1/10)$ and coupling the processes in dimension $N=1,2.\dots$ arbitrarily, we will show that almost surely,
    \begin{equation}
    \label{eq:CK-need}
    \sup_{s,t\in [0,T]}|C_N(s,t)-C(s,t)|\leq \delta
    \end{equation}
    holds for all but finitely many $N$. This implies the desired result.

    For each $N$ couple $\bx_t$ with $\by_t^{(L)}$ using a shared Brownian motion $\bB_t$. For any $L\geq L_0(p,\beta,\delta,T)$ as in Lemma~\ref{lem:lipschitz-approx}, with probability $1-e^{-cN}$ (in particular for all but finitely many $N$):
    \[
    \sup_{t\in [0,T]}
    \|\bx_t-\by_t^{(L)}\|\leq \frac{\delta\sqrt{N}}{10}.
    \]
    Recalling Definition~\ref{def:empirical-correlation}, this event implies that 
    \[
    \sup_{s,t\in [0,T]}
    |C_N(s,t)-C_N^{(L)}(s,t)|
    \leq \delta/3.
    \]
    Moreover the former convergence in Proposition~\ref{prop:soft-convergence} implies that for $L\geq L_1(p,\beta,\delta,T)$,
    \[
    \sup_{s,t\in [0,T]}
    |C^{(L)}(s,t)-C(s,t)|
    \leq \delta/3
    \]
    while the latter implies that for any fixed $L$, for all but finitely many $N$:
    \[
    \sup_{s,t\in [0,T]}
    |C_N^{(L)}(s,t)-C^{(L)}(s,t)|
    \leq \delta/3
    \]
    Finally fixing $L\geq \max(L_0,L_1)$ and combining the previous three displays using the triangle inequality, we deduce that \eqref{eq:CK-need} holds for all but finitely many $N$. This completes the proof.
\end{proof}

\paragraph{Acknowledgement}

Thanks to Gerard Ben Arous, Alex Damian, Reza Gheissari, Brice Huang, Aukosh Jagannath, Mehtaab Sawhney, Pierfrancesco Urbani, and the anonymous referees for helpful discussions and feedback.

\footnotesize
\bibliographystyle{alpha}
\bibliography{bib}

\newcommand{\etalchar}[1]{$^{#1}$}
\begin{thebibliography}{BCKM98}

\bibitem[AB13]{auffinger2013complexity}
Antonio Auffinger and G{\'e}rard Ben\space{}Arous.
\newblock Complexity of random smooth functions on the high-dimensional sphere.
\newblock {\em Ann. Probab.}, 41(6):4214--4247, 2013.

\bibitem[AB{\v{C}}13]{auffinger2013random}
Antonio Auffinger, G{\'e}rard Ben\space{}Arous, and Ji{\v{r}}{\'\i}
  {\v{C}}ern{\`y}.
\newblock Random matrices and complexity of spin glasses.
\newblock {\em Comm. Pure Appl. Math.}, 66(2):165--201, 2013.

\bibitem[ABXY22]{adhikari2022spectral}
Arka Adhikari, Christian Brennecke, Changji Xu, and Horng-Tzer Yau.
\newblock {Spectral Gap Estimates for Mixed $p$-Spin Models at High
  Temperature}.
\newblock {\em arXiv:2208.07844}, 2022.

\bibitem[AC17]{auffinger2017parisi}
Antonio Auffinger and Wei-Kuo Chen.
\newblock Parisi formula for the ground state energy in the mixed $p$-spin
  model.
\newblock {\em Ann. Probab.}, 45(6b):4617--4631, 2017.

\bibitem[AJK{\etalchar{+}}22]{anari2022entropic}
Nima Anari, Vishesh Jain, Frederic Koehler, Huy~Tuan Pham, and Thuy-Duong
  Vuong.
\newblock Entropic independence: optimal mixing of down-up random walks.
\newblock In {\em Proc. 54th STOC}, pages 1418--1430, 2022.

\bibitem[AMS21]{ams20}
Ahmed~El Alaoui, Andrea Montanari, and Mark Sellke.
\newblock Optimization of mean-field spin glasses.
\newblock {\em Ann. Probab.}, 49(6):2922--2960, 2021.

\bibitem[AMS23]{auffinger2022optimization}
Antonio Auffinger, Andrea Montanari, and Eliran Subag.
\newblock Optimization of random high-dimensional functions: Structure and
  algorithms.
\newblock In {\em Spin Glass Theory and Far Beyond: Replica Symmetry Breaking
  After 40 Years}, pages 609--633. World Scientific, 2023.

\bibitem[BB{\v{C}}08]{ben2008universality}
G{\'e}rard Ben\space{}Arous, Anton Bovier, and Ji{\v{r}}{\'\i} {\v{C}}ern{\`y}.
\newblock {Universality of the REM for dynamics of mean-field spin glasses}.
\newblock {\em Comm. Math. Phys.}, 282:663--695, 2008.

\bibitem[B{\v{C}}06]{ben2006dynamics}
G{\'e}rard Ben\space{}Arous and Ji{\v{r}}{\'\i} {\v{C}}ern{\`y}.
\newblock Dynamics of trap models.
\newblock In {\em Les Houches Summer School Proceedings}, volume~83, pages
  331--394. Elsevier BV, 2006.

\bibitem[B{\v{C}}08]{ben2008arcsine}
G{\'e}rard Ben\space{}Arous and Ji{\v{r}}{\'\i} {\v{C}}ern{\`y}.
\newblock The arcsine law as a universal aging scheme for trap models.
\newblock {\em Comm. Pure Appl. Math.: A Journal Issued by the Courant
  Institute of Mathematical Sciences}, 61(3):289--329, 2008.

\bibitem[BCKM98]{bouchaud1998out}
Jean-Philippe Bouchaud, Leticia~F Cugliandolo, Jorge Kurchan, and Marc
  M\'ezard.
\newblock Out of equilibrium dynamics in spin-glasses and other glassy systems.
\newblock {\em Spin glasses and random fields}, pages 161--223, 1998.

\bibitem[B{\v{C}}NS22]{belius2022triviality}
David Belius, Ji{\v{r}}{\'\i} {\v{C}}ern{\`y}, Shuta Nakajima, and Marius~A
  Schmidt.
\newblock Triviality of the geometry of mixed p-spin spherical hamiltonians
  with external field.
\newblock {\em J. Stat. Phys.}, 186(1):12, 2022.

\bibitem[BD95]{bouchaud1995aging}
J-P Bouchaud and David~S Dean.
\newblock {Aging on Parisi's tree}.
\newblock {\em Journal de Physique I}, 5(3):265--286, 1995.

\bibitem[BDG01]{arous2001aging}
G{\'e}rard Ben\space{}Arous, Amir Dembo, and Alice Guionnet.
\newblock Aging of spherical spin glasses.
\newblock {\em Probab. Theory Rel. Fields}, 120:1--67, 2001.

\bibitem[BDG06]{ben2006cugliandolo}
G{\'e}rard Ben\space{}Arous, Amir Dembo, and Alice Guionnet.
\newblock {Cugliandolo-Kurchan equations for dynamics of spin-glasses}.
\newblock {\em Probab. Theory Rel. Fields}, 136(4):619--660, 2006.

\bibitem[Ben02]{arous2002aging}
G{\'e}rard Ben\space{}Arous.
\newblock Aging and spin-glass dynamics.
\newblock In {\em Proceedings of the International Congress of Mathematicians},
  pages 3--14. Higher Ed. Press, 2002.

\bibitem[BG97]{arous1997symmetric}
G{\'e}rard Ben\space{}Arous and Alice Guionnet.
\newblock {Symmetric Langevin spin glass dynamics}.
\newblock {\em Ann. Probab.}, 25(3):1367--1422, 1997.

\bibitem[BG12]{ben2012universality}
G{\'e}rard Ben\space{}Arous and Onur G{\"u}n.
\newblock Universality and extremal aging for dynamics of spin glasses on
  subexponential time scales.
\newblock {\em Comm. Pure Appl. Math.}, 65(1):77--127, 2012.

\bibitem[BGJ20a]{ben2020algorithmic}
G{\'e}rard Ben\space{}Arous, Reza Gheissari, and Aukosh Jagannath.
\newblock {Algorithmic thresholds for tensor PCA}.
\newblock {\em Ann. Probab.}, 48(4):2052--2087, 2020.

\bibitem[BGJ20b]{ben2020bounding}
G{\'e}rard Ben\space{}Arous, Reza Gheissari, and Aukosh Jagannath.
\newblock Bounding flows for spherical spin glass dynamics.
\newblock {\em Comm. Math. Phys.}, 373:1011--1048, 2020.

\bibitem[BGJ21]{ben2021online}
G{\'e}rard Ben\space{}Arous, Reza Gheissari, and Aukosh Jagannath.
\newblock Online stochastic gradient descent on non-convex losses from
  high-dimensional inference.
\newblock {\em The Journal of Machine Learning Research}, 22(1):4788--4838,
  2021.

\bibitem[BGJ22]{ben2022high}
G{\'e}rard Ben\space{}Arous, Reza Gheissari, and Aukosh Jagannath.
\newblock {High-dimensional limit theorems for SGD: Effective dynamics and
  critical scaling}.
\newblock {\em Advances in Neural Information Processing Systems},
  35:25349--25362, 2022.

\bibitem[Bir99]{biroli1999dynamical}
Giulio Biroli.
\newblock {Dynamical TAP approach to mean field glassy systems}.
\newblock {\em Journal of Physics A: Mathematical and General}, 32(48):8365,
  1999.

\bibitem[CCM21]{celentano2021high}
Michael Celentano, Chen Cheng, and Andrea Montanari.
\newblock The high-dimensional asymptotics of first order methods with random
  data.
\newblock {\em arXiv:2112.07572}, 2021.

\bibitem[Che13]{chen2013aizenman}
Wei-Kuo Chen.
\newblock {The Aizenman-Sims-Starr scheme and Parisi Formula for Mixed $p$-spin
  Spherical Models}.
\newblock {\em Electronic Journal of Probability}, 18:1--14, 2013.

\bibitem[CHS93]{crisanti1993spherical}
Andrea Crisanti, Heinz Horner, and H~J Sommers.
\newblock {The spherical $p$-spin interaction spin-glass model: the dynamics}.
\newblock {\em Zeitschrift f{\"u}r Physik B Condensed Matter}, 92:257--271,
  1993.

\bibitem[CK93]{cugliandolo1993analytical}
Leticia~F Cugliandolo and Jorge Kurchan.
\newblock Analytical solution of the off-equilibrium dynamics of a long-range
  spin-glass model.
\newblock {\em Phys. Rev. Lett.}, 71(1):173, 1993.

\bibitem[Cug04]{cugliandolo2004course}
Leticia~F Cugliandolo.
\newblock Course 7: Dynamics of glassy systems.
\newblock In {\em Slow Relaxations and nonequilibrium dynamics in condensed
  matter: Les Houches Session LXXVII, 1-26 July, 2002}, pages 367--521.
  Springer, 2004.

\bibitem[Cug23]{cugliandolo2023recent}
Leticia~F Cugliandolo.
\newblock {Recent Applications of Dynamical Mean-Field Methods}.
\newblock {\em arXiv:2305.01229}, 2023.

\bibitem[DG21]{dembo2021diffusions}
Amir Dembo and Reza Gheissari.
\newblock {Diffusions interacting through a random matrix: universality via
  stochastic Taylor expansion}.
\newblock {\em Probab. Theory Rel. Fields}, 180:1057--1097, 2021.

\bibitem[DGM07]{dembo2007limiting}
Amir Dembo, Alice Guionnet, and Christian Mazza.
\newblock Limiting dynamics for spherical models of spin glasses at high
  temperature.
\newblock {\em J. Stat. Phys.}, 126:781--815, 2007.

\bibitem[DS20]{dembo2020dynamics}
Amir Dembo and Eliran Subag.
\newblock Dynamics for spherical spin glasses: disorder dependent initial
  conditions.
\newblock {\em J. Stat. Phys.}, 181:465--514, 2020.

\bibitem[EKZ22]{eldan2022spectral}
Ronen Eldan, Frederic Koehler, and Ofer Zeitouni.
\newblock {A spectral condition for spectral gap: fast mixing in
  high-temperature Ising models}.
\newblock {\em Probab. Theory Rel. Fields}, 182(3-4):1035--1051, 2022.

\bibitem[FFRT21]{folena2021gradient}
Giampaolo Folena, Silvio Franz, and Federico Ricci-Tersenghi.
\newblock Gradient descent dynamics in the mixed $p$-spin spherical model:
  finite-size simulations and comparison with mean-field integration.
\newblock {\em J. Stat. Mech.: Theory and Experiment}, 2021(3):033302, 2021.

\bibitem[Gam21]{gamarnik2021survey}
David Gamarnik.
\newblock The overlap gap property: A topological barrier to optimizing over
  random structures.
\newblock {\em Proceedings of the National Academy of Sciences}, 118(41), 2021.

\bibitem[GJ19]{gheissari2019spectral}
Reza Gheissari and Aukosh Jagannath.
\newblock On the spectral gap of spherical spin glass dynamics.
\newblock In {\em Annales de l'Institut Henri Poincar{\'e}, Probabilit{\'e}s et
  Statistiques}, volume~55, pages 756--776. Institut Henri Poincar{\'e}, 2019.

\bibitem[GJW20]{gamarnik2020optimization}
David Gamarnik, Aukosh Jagannath, and Alexander~S. Wein.
\newblock Low-degree hardness of random optimization problems.
\newblock In {\em Proc. 61st FOCS}, pages 131--140. IEEE, 2020.

\bibitem[GM05]{guionnet2005long}
Alice Guionnet and Christian Mazza.
\newblock {Long time behaviour of the solution to non-linear Kraichnan
  equations}.
\newblock {\em Probab. Theory Rel. Fields}, 131(4):493--518, 2005.

\bibitem[Gru96]{grunwald1996sanov}
Malte Grunwald.
\newblock {Sanov results for Glauber spin-glass dynamics}.
\newblock {\em Probab. Theory Rel. Fields}, 106(2):187--232, 1996.

\bibitem[GS14]{gamarnik2014limits}
David Gamarnik and Madhu Sudan.
\newblock Limits of local algorithms over sparse random graphs.
\newblock In {\em Proceedings of the 5th conference on Innovations in
  theoretical computer science}, pages 369--376. ACM, 2014.

\bibitem[GS17]{gamarnik2017performance}
David Gamarnik and Madhu Sudan.
\newblock Performance of sequential local algorithms for the random
  {NAE}-{$K$}-sat problem.
\newblock {\em SIAM Journal on Computing}, 46(2):590--619, 2017.

\bibitem[Gui97]{guionnet1997averaged}
Alice Guionnet.
\newblock Averaged and quenched propagation of chaos for spin glass dynamics.
\newblock {\em Probability Theory \& Related Fields}, 109(2), 1997.

\bibitem[Gui07]{guionnet2007dynamics}
Alice Guionnet.
\newblock Dynamics for spherical models of spin-glass and aging.
\newblock {\em Spin glasses}, pages 117--144, 2007.

\bibitem[HS21]{huang2021tight}
Brice Huang and Mark Sellke.
\newblock {Tight Lipschitz Hardness for Optimizing Mean Field Spin Glasses}.
\newblock {\em arXiv:2110.07847}, 2021.

\bibitem[HS23a]{huang2023algorithmic}
Brice Huang and Mark Sellke.
\newblock Algorithmic threshold for multi-species spherical spin glasses.
\newblock {\em arXiv:2303.12172}, 2023.

\bibitem[HS23b]{kac-rice-in-progress}
Brice Huang and Mark Sellke.
\newblock {Strong Topological Trivialization of Multi-Species Spherical Spin
  Glasses}.
\newblock {\em arXiv:2308.09677}, 2023.

\bibitem[Jag19]{jagannath2019dynamics}
Aukosh Jagannath.
\newblock Dynamics of mean field spin glasses on short and long timescales.
\newblock {\em J. Math. Phys.}, 60(8):083305, 2019.

\bibitem[JNG{\etalchar{+}}21]{jin2021nonconvex}
Chi Jin, Praneeth Netrapalli, Rong Ge, Sham~M Kakade, and Michael~I Jordan.
\newblock {On Nonconvex Optimization for Machine Learning: Gradients,
  Stochasticity, and Saddle Points}.
\newblock {\em Journal of the ACM (JACM)}, 68(2):1--29, 2021.

\bibitem[Kiv23]{kivimae2021ground}
Pax Kivimae.
\newblock The ground state energy and concentration of complexity in spherical
  bipartite models.
\newblock {\em Communications in Mathematical Physics}, 403(1):37--81, 2023.

\bibitem[Mac21]{maclaurin2021emergent}
James MacLaurin.
\newblock An emergent autonomous flow for mean-field spin glasses.
\newblock {\em Probab. Theory Rel. Fields}, 180(1-2):365--438, 2021.

\bibitem[McK24]{mckenna2021complexity}
Benjamin McKenna.
\newblock Complexity of bipartite spherical spin glasses.
\newblock In {\em Annales de l'Institut Henri Poincare (B) Probabilites et
  statistiques}, volume~60, pages 636--657. Institut Henri Poincar{\'e}, 2024.

\bibitem[Mon21]{montanari2021optimization}
Andrea Montanari.
\newblock {Optimization of the Sherrington--Kirkpatrick Hamiltonian}.
\newblock {\em SIAM Journal on Computing}, (0):FOCS19--1, 2021.

\bibitem[Pan13]{panchenko2013parisi}
Dmitry Panchenko.
\newblock {The Parisi ultrametricity conjecture}.
\newblock {\em Annals of Mathematics}, pages 383--393, 2013.

\bibitem[Pan14]{panchenko2014parisi}
Dmitry Panchenko.
\newblock {The Parisi formula for mixed $p$-spin models}.
\newblock {\em Ann. Probab.}, 42(3):946 -- 958, 2014.

\bibitem[Par79]{parisi1979infinite}
Giorgio Parisi.
\newblock Infinite number of order parameters for spin-glasses.
\newblock {\em Phys. Rev. Lett.}, 43(23):1754, 1979.

\bibitem[RY13]{revuz2013continuous}
Daniel Revuz and Marc Yor.
\newblock {\em Continuous martingales and Brownian motion}, volume 293.
\newblock Springer Science \& Business Media, 2013.

\bibitem[SAR18]{simchowitz2018tight}
Max Simchowitz, Ahmed~El Alaoui, and Benjamin Recht.
\newblock {Tight Query Complexity Lower Bounds for PCA via Finite Sample
  Deformed Wigner Law}.
\newblock In {\em Proc. 50th STOC}, pages 1249--1259, 2018.

\bibitem[Sel24]{sellke2021optimizing}
Mark Sellke.
\newblock Optimizing mean field spin glasses with external field.
\newblock {\em Electronic Journal of Probability}, 29:1--47, 2024.

\bibitem[SK75]{sherrington1975solvable}
David Sherrington and Scott Kirkpatrick.
\newblock Solvable model of a spin-glass.
\newblock {\em Phys. Rev. Lett.}, 35(26):1792, 1975.

\bibitem[Sub17a]{subag2017complexity}
Eliran Subag.
\newblock The complexity of spherical $p$-spin models—a second moment
  approach.
\newblock {\em Ann. Probab.}, 45(5):3385--3450, 2017.

\bibitem[Sub17b]{subag2017geometry}
Eliran Subag.
\newblock {The Geometry of the Gibbs Measure of Pure Spherical Spin Glasses}.
\newblock {\em Inventiones mathematicae}, 210(1):135--209, 2017.

\bibitem[Sub21]{subag2018following}
Eliran Subag.
\newblock {Following the Ground States of Full-RSB Spherical Spin Glasses}.
\newblock {\em Comm. Pure Appl. Math.}, 74(5):1021--1044, 2021.

\bibitem[Sub23]{subag2021tap2}
Eliran Subag.
\newblock {TAP approach for multispecies spherical spin glasses II: The free
  energy of the pure models}.
\newblock {\em Ann. Probab.}, 51(3):1004--1024, 2023.

\bibitem[SZ82]{sompolinsky1982relaxational}
Haim Sompolinsky and Annette Zippelius.
\newblock {Relaxational dynamics of the Edwards-Anderson model and the
  mean-field theory of spin-glasses}.
\newblock {\em Physical Review B}, 25(11):6860, 1982.

\bibitem[SZ17]{subag2017extremal}
Eliran Subag and Ofer Zeitouni.
\newblock The extremal process of critical points of the pure $p$-spin
  spherical spin glass model.
\newblock {\em Probab. Theory Rel. Fields}, 168(3-4):773--820, 2017.

\bibitem[SZ21]{subag2021concentration}
Eliran Subag and Ofer Zeitouni.
\newblock Concentration of the complexity of spherical pure $p$-spin models at
  arbitrary energies.
\newblock {\em J. Math. Phys.}, 62(12):123301, 2021.

\bibitem[Tal06a]{talagrand2006spherical}
Michel Talagrand.
\newblock Free energy of the spherical mean field model.
\newblock {\em Probab. Theory Rel. Fields}, 134(3):339--382, 2006.

\bibitem[Tal06b]{talagrand2006parisi}
Michel Talagrand.
\newblock {The Parisi formula}.
\newblock {\em Annals of Mathematics}, pages 221--263, 2006.

\bibitem[Zam08]{zamfir2008limiting}
Pompiliu~Manuel Zamfir.
\newblock Limiting dynamics for spherical models of spin glasses with magnetic
  field.
\newblock {\em arXiv:0806.3519}, 2008.

\bibitem[ZLC17]{zhang2017hitting}
Yuchen Zhang, Percy Liang, and Moses Charikar.
\newblock {A Hitting Time Analysis of Stochastic Gradient Langevin Dynamics}.
\newblock In {\em Conference on Learning Theory}, pages 1980--2022. PMLR, 2017.

\end{thebibliography}

\end{document}